\documentclass[11pt]{article}
\usepackage[T1]{fontenc}
\usepackage[margin = 1in]{geometry}
\usepackage{textcomp}
\usepackage{graphicx}

\usepackage[
backend=bibtex,
style=alphabetic,
sorting=nty,
maxbibnames=99
]{biblatex}
\renewbibmacro{in:}{}
\DeclareFieldFormat*{title}{#1}
\usepackage{amsmath,amssymb,amsthm,hyperref,color}
\usepackage{float}
\usepackage[font=small]{caption}
\usepackage[font=footnotesize]{subcaption}
\usepackage[inline]{enumitem}
\usepackage{tikz}
\hypersetup{colorlinks=true,citecolor=blue, linkcolor=blue,
urlcolor=blue}

\newtheorem{theorem}{Theorem}
\newtheorem{lemma}[theorem]{Lemma}

\newtheorem{definition}[theorem]{Definition}

\newtheorem*{lemintermediatePotts}{Lemma~\ref{lem:intermediatePotts}}
\newtheorem*{lemintermediateIsing}{Lemma~\ref{lem:intermediateIsing}}
\newtheorem*{lemPottsconstruction}{Lemma~\ref{lem:Pottsconstruction}}
\newtheorem*{lemspinconstruction}{Lemma~\ref{lem:2spinconstruction}}
\newtheorem*{lemintergadget}{Lemma~\ref{lem:intergadget}}

\renewcommand\arg{\text{arg}}

\def\calI{\mathcal{I}}
\def\Gc{\mathcal{G}}
\def\Tc{\mathcal{T}}
\def\Yc{\mathcal{Y}}
\def\Pc{\mathcal{P}}
\def\Ac{\mathcal{A}}

\def\oc{o}
\def\mc{m}
\def\equal{0}
\def\nequal{1}

\def\calG{\mathcal{G}}
\def\bit{\mathrm{bits}}

\def\two{\{0,1\}}

\newcommand{\pl}{\textup{\texttt{+}}}
\newcommand{\mi}{\textup{\texttt{-}}}
\newcommand{\plm}{\textup{\texttt{\textpm}}}
\newcommand{\mip}{\textup{\texttt{\raisebox{.1em}{%
	\reflectbox{\rotatebox[origin=c]{180}{\textpm}}}}}}

\def\lambdab{\ensuremath{\boldsymbol{\lambda}}}
\def\betab{\ensuremath{\boldsymbol{\beta}}}

\def\twospin{\mathrm{2spin}}
\def\Potts{\mathrm{Potts}}

\def\Mb{\ensuremath{\mathbf{M}}}
\def\Mc{\ensuremath{\mathcal{M}}}
\def\Oc{\ensuremath{\mathcal{O}}}
\def\Eb{\ensuremath{\mathbf{E}}}
\def\Ec{\ensuremath{\mathcal{E}}}

\def\Sc{\ensuremath{\mathcal{S}}}
\def\Uc{\ensuremath{\mathcal{N}}}

\newcommand{\EST}{\mathsf{EST}}

\def\Gc{\ensuremath{\mathcal{G}}}

\def\Susc{\#\ensuremath{\mathsf{Susc}}}
\def\SuscCubic{\#\ensuremath{\mathsf{SuscCubic}}}
\def\PottsCubic{\#\ensuremath{\mathsf{PottsCubic}}}
\def\IsingCubic{\#\ensuremath{\mathsf{IsingCubic}}}
\def\MagIsing{\#\ensuremath{\mathsf{MagnetIsingCubic}}}
\def\Obser{\#\ensuremath{\mathsf{Observable2Spin}}}

\newcommand{\fptas}{\mathsf{FPTAS}}
\newcommand{\fpras}{\mathsf{FPRAS}}

\newcommand{\eps}{\epsilon}

\makeatletter
\def\prob#1#2#3{\goodbreak\begin{list}{}{\labelwidth\z@ \itemindent-\leftmargin
                        \itemsep\z@  \topsep6\p@\@plus6\p@
                        \let\makelabel\descriptionlabel}
                \item[\it Name]#1
               \item[\it Instance]                #2
                \item[\it Output]#3
                \end{list}}
\makeatother

\title{Approximating observables is as hard as counting}
\author{
Andreas Galanis\thanks{
  Department of Computer Science, University of Oxford, Wolfson Building, Parks Road, Oxford, OX1~3QD, UK.}
\and
Daniel \v{S}tefankovi\v{c}\thanks{Department of Computer Science, University of Rochester, USA. Research supported in part by NSF grant CCF-1563757.}
\and
Eric Vigoda\thanks{Department of Computer Science, University of California, Santa Barbara, CA 93106, USA.  
			Email: vigoda@ucsb.edu.
		Research supported in part by NSF grant CCF-2007022.}
}
\date{June 23, 2022}

\begin{document}

\maketitle
\thispagestyle{empty}
\begin{abstract}
We study the computational complexity of estimating local observables
for Gibbs distributions. A simple combinatorial example is
the average size of an independent set in a graph. In a recent work, we 
established NP-hardness of approximating the average
size of an independent set utilizing hardness of the corresponding
optimization problem and the related phase transition behavior. Here, we instead consider 
settings where the underlying optimization problem is easily solvable.  
 Our main contribution is to classify the 
complexity of approximating a  wide class of observables 
via  a generic reduction from approximate counting to the problem of estimating local
observables. The key idea is to  use the observables to interpolate the
counting problem.

Using this new approach, we are able to study observables on bipartite graphs
where the underlying optimization problem is easy but the counting problem is believed to be hard.
The most-well studied class of graphs that was excluded from previous hardness results were bipartite graphs.
We establish hardness for estimating the average size of the independent set in bipartite graphs
of maximum degree 6; more generally, we show tight hardness results for 
general vertex-edge observables for antiferromagnetic 2-spin systems
on bipartite graphs. Our techniques go beyond 2-spin systems, and for 
the ferromagnetic Potts model we establish hardness of approximating the number of monochromatic edges in the same region as known hardness of approximate counting results.
\end{abstract}

\newpage

\clearpage
 \setcounter{page}{1}

\section{Introduction}

Can we efficiently estimate the average size of an independent set in an input graph $G=(V,E)$?
Moreover, can we do so without utilizing a sampling algorithm for generating a random independent set?

In this paper, for a broad class of problems captured by Gibbs distributions,   we address the relationship between the computational complexity of approximating local 
observables  (such as estimating the average size of an independent set)  and the computational complexity of
approximating the partition function (such as estimating the total number of independent sets).  It is a standard technique in the area to reduce estimating observables to approximate counting, by first implementing an approximate sampler and then using an unbiased estimator of the desired observable. The focus of this paper is the converse, where there is no previously known technique to answer the following question: does an algorithm for local observables yield an algorithm for the partition function? We prove, in a broad setting, that these two genres of problems are 
computationally equivalent.

Previous work of \cite{averages} only achieved this indirectly, where the hardness of approximating local observables was established utilizing the hardness of \textsc{MaxCut} (in fact, only for a certain observable, called magnetization, see below for definitions). Here, we show a direct reduction from the observable problem to the partition-function problem, relating therefore more crisply the two problems. This allows us to obtain hardness results in several new regimes (in particular, not covered by \cite{averages}) where the counting problem is hard but there is no underlying hard optimization problem. 

An interesting setting to highlight the usefulness of our reduction is {\em bipartite} independent sets.  In this example there is no corresponding hard optimization problem (as the maximum independent set problem is poly-time solvable in bipartite graphs), and hence to prove hardness we need to utilize hardness of approximate counting results.   
Another pertinent example for our results are attractive graphical models, these are equivalent to ferromagnetic spin systems in statistical physics.  The simplest case is the ferromagnetic Ising model and its generalization known as the Potts model. In the Ising/Potts model on a graph (see Section~\ref{sec:pottsintro} for more precise definitions), the configurations of the model are the collection of labellings $\sigma$ of the vertices with $q$ spins (colours), each weighted as $\beta^{m(\sigma)}$ where $m(\sigma)$ is the number of monochromatic edges and $\beta$ is a parameter $>1$ (so that labellings with many monochromatic edges are favored). Because of the attractiveness assumption that $\beta>1$, once again, there is no corresponding hard optimization problem for this problem (contrast this with the case $\beta<1$ where the largest weight labellings have the smallest number of monochromatic edges). Nevertheless, using our new reduction, we show that hardness of the associated approximate counting problem implies hardness of estimating the (weighted) average of the monochromatic  edges in the Potts model.  

Our two illustrative examples, the average size of an independent set and the number of monochromatic edges in the Ising/Potts model, are instances of a local observable in statistical physics; specifically they correspond to the magnetization and susceptibility, respectively.  The behavior of observables is fundamental to the study of phase transitions, e.g., see~\cite{Baxter,SS}.

We begin  giving more precise definitions for our initial example of {\em bipartite} independent sets, before considering the ferromagnetic Potts model, and finally generalizing to arbitrary local observables in general 2-spin systems.
For a graph $G=(V,E)$  let $\calI_G$ denote the set of independent sets (of all sizes) of $G$, and let~$\mu:=\mu_G$ 
denote the uniform distribution over $\calI_G$.  
Denote the average independent set size by $\Mc(G)=\Eb_{\sigma\sim \mu}\big[|\sigma|\big]$; this corresponds
to the magnetization in statistical physics (and hence the choice of notation~$\Mc$). We say that an algorithm is an $\fpras$ for the 
average independent set size if 
given a graph $G=(V,E)$ and parameters $\epsilon,\delta>0$, 
the algorithm outputs an estimate $\EST$ which is within a multiplicative factor $(1\pm\epsilon)$ of the desired quantity $\Mc(G)$, 
with probability $\geq 1-\delta$, and runs in time
$\mathrm{poly}(|V|,1/\epsilon,\log(1/\delta))$.
One can also consider an $\fpras$ 
for estimating $|\Omega|$, the number of independent sets of the input graph $G$;
we refer to this as an efficient approximate counting algorithm.

It is a classical result \cite{JERRUM1986169} that an efficient
approximate counting algorithm is polynomial-time interreducible with an efficient algorithm for approximate
sampling from~$\mu$. In turn, efficiently estimating the average independent set size of a graph $G$ 
is easily reduced to approximate sampling from the uniform 
distribution~$\mu_G$. The challenging aspect, and the focus of this paper, is the \emph{reverse} implication. Can we estimate the typical size of an independent set without utilizing an approximate sampling algorithm? We will show it is \emph{not} possible, i.e., hardness of approximate counting (and hence approximate sampling) implies hardness of estimating the average independent set size.

For graphs of maximum degree $5$, Weitz~\cite{Weitz} presented an $\fptas$ for approximating the number of independent sets, which yields an efficient approximate sampling scheme; note an $\fptas$ is the deterministic analog of an $\fpras$, i.e., it achieves $\delta=0$.
Very recently, Chen et al.~\cite{CLV} proved that the simple MCMC algorithm known as the Gibbs sampler
(or Glauber dynamics) has $O(n\log{n})$ mixing time for this same class of graphs of maximum degree $5$.
Hence, one immediately obtains an $\fpras$ for the average independent set size~$\Mc(G)$.

On the other side, for graphs of maximum degree $6$, Sly~\cite{Sly10} proved that approximating the number
of independent sets is NP-hard, by a reduction from max-cut.  Schulman et al.~\cite{SSS} showed \#P-hardness for exact computation of the average independent set size.  Moreover, recent work of Galanis et al.~\cite{averages}
shows that approximating the average independent-set size is also NP-hard for graphs of maximum degree $6$.
The proof of~\cite{averages} does not directly relate approximate counting and estimating the average independent
set size; instead~\cite{averages} also shows a (more sophisticated) reduction from max-cut and utilizes the associated gadgets used in
Sly's inapproximability result~\cite{Sly10}.

This begets the question: are these problems still intractable when restricted to {\em bipartite graphs}?
For bipartite graphs there is no longer a hard optimization problem, such as max-cut, that one can utilize as a starting point for 
a hardness reduction.  However, approximately counting independent sets is  considered to be intractable 
on bipartite graphs of maximum degree $6$; in particular, it is  is \#BIS-hard~\cite{CAI} where \#BIS refers
to the problem of approximately counting independent sets on general bipartite graphs (with potentially unbounded degree).  
There are now a multitude of approximate counting problems which share the same \#BIS-hardness status or are even \#BIS-equivalent, e.g., see~\cite{DGGJ,CAI,Potts,GGJ}.

We present a general approach for reducing approximate counting to approximating averages.  This
yields hardness for approximating the average independent-set size in {\em bipartite graphs} of maximum degree 6.
\begin{theorem}\label{thm:main-IS}
Let $\Delta\geq 6$ be an integer.  There is no $\fpras$ for the average independent-set size on {\em bipartite} graphs of maximum degree $\Delta$ unless 
$\#\mathrm{BIS}$ admits an $\fpras$.
\end{theorem}
Note that the $\#\mathrm{BIS}$-hardness result of Theorem~\ref{thm:main-IS}  gives a  weaker guarantee than those shown in~\cite{averages} where they obtain in some cases constant-factor inapproximability results (using the constant-factor NP-hardness of the optimization problem). This difference is inherent with the $\#\mathrm{BIS}$-hardness assumption, i.e., that there is no $\fpras$ for $\#\mathrm{BIS}$.
Moreover, an algorithm which approximates $\#\mathrm{BIS}$ within any $\mathrm{poly}(n)$-factor implies an $\fpras$, and obtaining constant-factor inapproximability results for magnetization on bipartite graphs would require (among other things) hardness of $\#\mathrm{BIS}$ within an exponential-factor.

Our results extend to the hard-core model on weighted independent sets, and to general 2-spin antiferromagnetic models.  
These more general results are detailed in Section~\ref{sec:general-intro}.

\subsection{Ferromagnetic Potts Model}\label{sec:pottsintro}
Ferromagnetic spin systems, which are equivalent to attractive undirected graphical models, are an interesting class of models to illustrate our 
new proof technique on.  In ferromagnetic models there is no hard optimization problem as the maximum likelihood configurations are 
trivial assignments (setting all vertices to the same spin/label).  Consequently, to obtain hardness results for computing
averages in ferromagnetic models we need to work directly from hardness of approximate counting results, which we can do using our 
new approach.  

The most well-studied examples of ferromagnetic models are the Ising and Potts models.  
Given a graph $G$ and an integer $q\geq 2$, configurations of the Ising/Potts model are
the collection $\Omega$ of assignments $\sigma:V(G)\rightarrow [q]$ where $[q]=\{1,\hdots,q\}$ are the labels of the $q$ spins.
The case $q=2$ corresponds to the Ising model
and the case $q\geq 3$ is the Potts model.  The  models are parameterised by an edge activity\footnote{ We remark that $\beta$ is usually used to denote the so-called inverse temperature of the Potts model; here to have consistent notation with general 2-spin systems presented in Section~\ref{sec:general-intro} 
we take $\beta$ to be the exponent of the inverse temperature.} $\beta>0$.
 The weight of an assignment $\sigma$ is defined as
$w_{G;q,\beta}(\sigma)=\beta^{m_G(\sigma)}$ where $m_G(\sigma)=|\{(u,v)\in E: \sigma(u)=\sigma(v)\}|$ is the number of edges which are monochromatic in $\sigma$.  
Finally, the Gibbs distribution is defined as $\mu_{G;q,\beta}(\sigma) = w_{G;q,\beta}(\sigma)/Z_{G;q,\beta}$ where the normalising factor  $Z_{G;q,\beta}:=\sum_{\tau:V(G)\rightarrow [q]} w(\tau)$ is the partition function. In this paper, we restrict attention to the case $\beta>1$ 
which is the ferromagnetic (attractive) case, and hence
the most likely configurations are the $q$ monochromatic configurations  (all vertices are assigned the same spin).

For the Ising and Potts models, the analog of the average independent set size is the average number of vertices assigned spin $1$.  
This quantity $\Mc_{q,\beta}(G)$, known as the magnetization, is trivial in these cases since, due to the Ising/Potts models symmetry among spins, it holds that $\Mc_{q,\beta}(G)=n/q$.  
The simplest and most natural   observable to consider is the average number of monochromatic edges under the Potts distribution, i.e., the quantity 
\[\Sc_{q,\beta}(G):=\Eb_{\sigma\sim \mu_{G;q,\beta}}[ m_G(\sigma)]\] 
which is known as the \emph{susceptibility}.  
Sinclair and Srivastava~\cite{SinclairSrivastava} showed that exact computation of the susceptibility in the ferromagnetic Ising model is \#P-hard.

For the Ising model a classical result of Jerrum and Sinclair~\cite{JS:ising} presents an efficient sampling scheme for all $G$, all $\beta$.  
This yields an efficient algorithm for approximating averages in the Ising model (this holds for any local observables as defined subsequently in Section~\ref{sec:general-intro}).  In contrast for the Potts model (for any $q\geq 3$) approximating the partition function becomes 
computationally intractable for large $\beta$ as we detail below.

The Potts model with $q\geq 3$ spins undergoes a computational phase transition on 
bipartite graphs of maximum degree~$\Delta$ at
the following critical point $\beta_c(q,\Delta)=\frac{q-2}{(q-1)^{1-2/\Delta}-1}$.  
In \cite{PottsBIS} it was established that for all $q,\Delta\geq 3$ and $\beta>\beta_c(q,\Delta)$ approximating the partition function of the ferromagnetic Potts model  is \#BIS-hard on bipartite graphs of maximum degree $\Delta$.  Using our general counting-to-observables reduction we show that
 approximating the average number of monochromatic edges under the Potts distribution is as hard as 
approximating the partition function for the ferromagnetic Potts model.

\begin{theorem}\label{thm:main1}
Let $q,\Delta\geq 3$ be integers and $\beta> \beta_c(q,\Delta)$. 
There is no $\fpras$ for the susceptibility in the $q$-state Potts model on bipartite graphs of maximum degree $\Delta$, 
unless $\#\mathrm{BIS}$ admits an $\fpras$.
\end{theorem}

\subsection{General 2-spin systems}\label{sec:general-intro}
Theorem~\ref{thm:main-IS} for independent sets is a special case of a general result for arbitrary 2-spin antiferromagnetic systems.  
Such spin systems are parameterized by three parameters, $\beta,\gamma$ and~$\lambda$; the first two are edge activities and control the strength of the spin interactions between neighboring vertices, and the third is a vertex activity (a.k.a. external field) that favors one spin over the other.

More precisely, for a graph $G=(V,E)$, $\beta,\gamma\geq 0$ which are not both equal to zero and $\lambda>0$, let $\mu_{G;\beta,\gamma,\lambda}$ denote the Gibbs distribution on $G$ with edge activities $\beta,\gamma$ and external field $\lambda$, i.e., for $\sigma:V\rightarrow \{0,1\}$ we have 
\[\mu_{G;\beta,\gamma,\lambda}(\sigma)=\frac{\lambda^{|\sigma|}\beta^{m_{0}(\sigma)}\gamma^{m_{1}(\sigma)}}{Z_{G;\beta,\gamma,\lambda}},\]
where  $|\sigma|$ is the number of vertices with spin $1$, and  $m_{0}(\sigma), m_{1}(\sigma)$ denote the number of edges in $G$ whose endpoints are assigned under $\sigma$ the pair of spins $(0,0)$ and $(1,1)$, respectively.

The parameter pair $(\beta,\gamma)$ is called \emph{antiferromagnetic} if $\beta\gamma\in[0,1)$ and at least one of $\beta,\gamma$ is non-zero, and it is called ferromagnetic, otherwise.   Note that the hard-core model on independent sets weighted by $\lambda>0$ is the case $\beta=1,\gamma=0$ (under the convention that $0^0\equiv 1$).  Our earlier example of unweighted independent sets corresponds to the hard-core model with $\lambda=1$.
The antiferromagnetic Ising model is the special case $0<\beta=\gamma<1$.

Our results apply to general ``vertex-edge observables'' defined as follows.
\begin{definition}\label{def:observable}
Let $(\beta,\gamma)$ be antiferromagnetic and $\lambda>0$. For real numbers $a,b,c$, the \emph{$(a,b,c)$ vertex-edge observable} of a graph $G$ in the 2-spin system corresponding to $(\beta,\gamma,\lambda)$ is  given by
\[\Oc_{\beta,\gamma,\lambda}(G)=\Eb_{\sigma\sim \mu_{G;\beta,\gamma,\lambda}}\big[\oc_G(\sigma)\big], \mbox{ where } \oc_G(\sigma)=a|\sigma|+b m_0(\sigma)+c m_1(\sigma).\]
The observable is \emph{trivial on general graphs}  if any of the following hold: (i) $a=b=c=0$, (ii) $\beta=0$ and $a=c=0$, (iii) $\gamma=0$ and $a=b=0$, (iv) $\beta=\gamma$, $\lambda=1$ and $b+c=0$. We say that the observable is \emph{trivial on bipartite graphs}  if either any of the above hold, or $\beta=\gamma$ and $\lambda=1$.  Otherwise, we say that the observable is \emph{non-trivial}.

Notice that by setting 
 $(a,b,c)=(1,0,0)$ we obtain the magnetization $\Mc_{\beta,\gamma,\lambda}(G)$, which in the special case of the hard-core model with $\lambda=1$ is the average size of an independent set.
Furthermore, by setting $(a,b,c)=(0,1,1)$ we obtain the susceptibility, denoted by $\Sc_{\beta,\gamma,\lambda}(G)$, which is the average number of monochromatic edges. 
\end{definition}
The terminology ``trivial'' is applied liberally here and meant to convey that there is an efficient algorithm for the relevant parameters.  In particular, while cases (i)-(iii) are degenerate, case (iv) corresponds to the Ising model without an external field.  A classical (and highly non-trivial) result of Jerrum and Sinclair~\cite{JS:ising} presented an $\fpras$ for the ferromagnetic Ising model on any graph, any $\beta>1$.  Moreover, for bipartite graphs, the subcase $\beta<1$ (antiferromagnetic) can be reduced to an 
equivalent $\beta>1$ ferromagnetic system (by flipping the spins on one side of the bipartite graph).

We next define the range of parameters $(\beta,\gamma,\lambda)$ where our inapproximability results for vertex-edge observables apply; these are precisely the parameters where the hard-core and the antiferromagnetic Ising models exhibit non-uniqueness on the infinite $\Delta$-regular tree (for general 2-spin systems this threshold corresponds to
what is known as up-to-$\Delta$ non-uniqueness, which captures the computational phase transition).
\begin{definition}\label{def:nonuniq}
Let $\Delta\geq 3$ be an integer. We let $\Uc_{\Delta}$ be the set of $(\beta,\gamma,\lambda)$ such that $(\beta,\gamma)$ is antiferromagnetic, and the (unique) fixpoint $x^*>0$ of the function $f(x)=\frac{1}{\lambda} \big(\frac{ \beta x+1}{x+\gamma}\big)^{\Delta-1}$ satisfies $|f'(x^*)|>1$. The region $\Uc_{\Delta}$ is known as the \emph{non-uniqueness region}.
\end{definition}

Note there is an efficient sampling/counting algorithm for graphs of maximum degree $\Delta$, roughly\footnote{More precisely, the (strict) uniqueness region is defined as those $\beta,\gamma,\lambda)$ where the fixpoint $x^*$ in Definition~\ref{def:nonuniq} satisfies the (strict) inequality $|f'(x^*)|<1$. For certain monotonicity reasons, the algorithm for max-degree $\Delta$ graphs demands that $(\beta,\gamma,\lambda)$ lie in the intersection of these uniqueness regions for all degrees $d\leq \Delta$, see~\cite{LLY} for details.}, for $(\beta,\gamma,\lambda)$ outside the parameter region $\Uc_{\Delta}$ ~\cite{LLY,CLV}. Inside $\Uc_{\Delta}$, it is NP-hard to approximate the partition function on graphs of maximum degree~$\Delta$~\cite{SlySun}
and it is \#BIS-hard to approximate the partition function on \emph{bipartite} graphs of maximum degree~$\Delta$~\cite{CAI}.
We prove that it is also hard to compute any non-trivial vertex-edge observable in exactly this same region where the corresponding counting problem is hard.

\begin{theorem}\label{thm:main2}
Let $\Delta\geq 3$ be an integer and $(\beta,\gamma,\lambda)\in \Uc_{\Delta}$. Then, for any vertex-edge  observable  that is non-trivial on bipartite graphs, there is no $\fpras$ on bipartite graphs of maximum degree $\Delta$ 
unless $\#\mathrm{BIS}$ admits an $\fpras$.
\end{theorem}

We stress that the above result holds for bipartite graphs.  The previous work of~Galanis et al.~\cite{averages} showed hardness for general antiferromagnetic 2-spin systems in the same non-uniqueness region but on general graphs, only for the magnetization, and only achieved the stronger 
constant-factor hardness for a dense set of $\lambda$.

\subsection{Outline of the rest of the paper}

We begin by establishing Theorem~\ref{thm:main1} for hardness of approximating the susceptibility for the ferromagnetic Potts model, see 
Section~\ref{sec:proofthm}.  
We then present the refinements to establish Theorems~\ref{thm:main2} for general 2-spin antiferromagnetic systems in Section~\ref{sec:2spin};
Theorem~\ref{thm:main-IS} follows as a corollary of Theorem~\ref{thm:main2}.

We conclude this introductory section by mentioning a related line of works \cite{Will1,Will2,Will3} which consider the computational problems of sampling/counting configurations with a fixed magnetization value; their techniques are different than ours, and the principles underlying their hardness results relate to extremal properties/graphs of the magnetization.

\section{Hardness of susceptibility for the ferromagnetic Potts model}\label{sec:proofthm}
Let $q,\Delta\geq 3$ be integers and $\beta>\beta_c(q,\Delta)$. To prove Theorem~\ref{thm:main2}, we will assume the existence of an $\fpras$ for the  susceptibility of Potts with parameters $q,\beta$ on maximum degree $\Delta$ graphs and show how to obtain an $\fpras$ for the partition function of the Potts model with parameters $q,\beta^*$ on bipartite graphs of maximum degree 3 for some $\beta^*>\beta_c(q,3)$; the latter problem is \#BIS-hard by \cite{PottsBIS}. 

To aid the presentation, it will be convenient to consider the following computational problems and use the notion of AP-reduction between counting problems \cite{DGGJ}; roughly, for two problems $A,B$,  the notation $A\leq_{\mathrm{AP}}B$ means that the existence of an FPRAS for $B$ implies the existence of an FPRAS for $A$. In the first computational problem that will be relevant, the parameters are $q,\beta,\Delta$ as detailed below.
\prob{$\Susc(q,\beta,\Delta)$.} 
{ A bipartite graph $G$ with max degree $\Delta$.} 
 { The susceptibility on $G$ with parameters $q,\beta$, i.e., the value $\Sc_{q,\beta}(G)$.}

In the second, the parameter is going to be just $q$; note that the problem allows the edge activity to be part of the input.
\prob{$\SuscCubic(q)$.} 
{ A cubic bipartite graph $H$, and a rational edge activity $\hat\beta\geq 1$.} 
 { The susceptibility  on $H$ with parameters $q,\hat\beta$, i.e., the value $\Sc_{q,\hat\beta}(H)$.}

The key ingredient underpinning our proof approach is captured by the following lemma, whose proof is given in Section~\ref{sec:rgtgtgrtte}. Roughly, the lemma asserts that, despite the fact that the parameter $\beta$ is fixed, with appropriate gadget constructions we can ``shift'' it in a fine-tuned way to any desired $\hat\beta$. In turn, this allows us to do an appropriate integration of the observable (viewed as a function of the parameter $\hat\beta$) to recover the partition function of a \#BIS-hard problem; we will refer loosely to this integration technique as \emph{interpolation}.
\newcommand{\statelemintermediatePotts}{Let $q,\Delta\geq 3$  be integers and $\beta>\beta_c(q,\Delta)$ be an arbitrary real. Then,  
\[\SuscCubic(q)\leq_{\mathrm{AP}} \Susc(q,\beta,\Delta).\]}
\begin{lemma}\label{lem:intermediatePotts}
\statelemintermediatePotts
\end{lemma}
Before proceeding with outlining the proof of the key Lemma~\ref{lem:intermediatePotts}, we first present the interpolation-scheme idea that allows us to conclude Theorem~\ref{thm:main1} from Lemma~\ref{lem:intermediatePotts}.
\begin{proof}[Proof of Theorem~\ref{thm:main1} (assuming Lemma~\ref{lem:intermediatePotts})]
Let $\beta^*>\beta_c(q,3)$ be an arbitrary rational number and consider the problem $\PottsCubic(q,\beta^*)$, i.e., the problem of approximating the partition function $Z_{G;q,\beta^*}$ on cubic bipartite graphs $G$.  From \cite[Theorem 3]{PottsBIS}, we have that $\PottsCubic(q,\beta^*)$ is \#BIS-hard. From Lemma~\ref{lem:intermediatePotts}, we have that for $\beta>\beta_p(q,\Delta)$ it holds that $\SuscCubic(q)\leq_{\mathrm{AP}} \Susc(q,\beta,\Delta)$, so to prove the theorem it suffices to show that $\PottsCubic(q,\beta^*)\leq_{\mathrm{AP}}  \SuscCubic(q)$.

Let $G$ be an instance of $\PottsCubic(q,\beta^*)$ with $n$ vertices and $m$ edges, and $\epsilon>0$ be the desired relative error  that we want to approximate $Z_{G;q,\beta^*}$. Since $\frac{\partial \log Z_{G;q,\hat\beta}}{\partial \hat\beta}=\tfrac{1}{\hat\beta}\Sc_{q,\hat\beta}(G)$, we have
\begin{equation}\label{eq:interpolation}
\log Z_{G;q,\beta^*}=\int^{\beta^*}_{1}\, \tfrac{1}{\hat\beta}\Sc_{q,\hat\beta}(G) \mathrm{d}\hat\beta.
\end{equation}
Let $M=\lceil (10q\hat\beta m/\epsilon)^4\rceil$ and for $i=0,1,\hdots, M$, consider the sequence of edge parameters $\hat\beta_i=1+i\tfrac{\beta^*-1}{M}$. It is a standard fact that the function $\log Z_{G;\hat\beta}$ is convex with respect to $\hat\beta$ (the second derivative is equal to the variance of the number of monochromatic edges) and therefore the function $\tfrac{1}{\hat\beta}\Sc_{q,\hat\beta}(G)$ is increasing. Therefore, from the standard technique of approximating integrals with rectangles, we obtain from \eqref{eq:interpolation} that
\[\frac{1}{M}\sum^{M-1}_{i=0}\tfrac{\Sc_{G;q,\hat\beta_i}}{\hat\beta_i}\leq \log Z_{G;\beta^*}\leq \frac{1}{M}\sum^{M}_{i=1}\tfrac{\Sc_{q,\hat\beta_i}(G)}{\hat\beta_i}.\]
Using the bound $m/q\leq \Sc_{q,\hat\beta}(G)\leq m$ that holds for all $\hat\beta\geq 1$, we obtain that
\[ \log Z_{G;q,\beta^*}=\big(1\pm\frac{\epsilon}{10}\big)\sum^M_{i=1}\tfrac{\Sc_{q,\hat\beta_i}(G)}{\hat\beta_i}.\]
Using the presumed oracle for $\SuscCubic(q)$ we can compute $\hat{S}_i$ such that $\hat{S}_i=(1\pm \tfrac{\epsilon}{10Mm})\Sc_{q,\hat\beta_i}(G)$ for $i\in [M]$, and therefore the quantity $\hat Z=\exp\big(\sum_{i\in [M]}  \tfrac{\Sc_{q,\hat\beta_i}(G)}{\hat\beta_i}\big)$
is a $(1\pm \epsilon)$-factor approximation to $Z_{G;q,\beta^*}$. This completes the AP-reduction, and therefore the proof as well.
\end{proof} 
In the rest of Section~\ref{sec:proofthm}, we focus on proving Lemma~\ref{lem:intermediatePotts}.

\subsection{Proof overview of Lemma~\ref{lem:intermediatePotts}}\label{sec:overview}

In this section, we give the proof overview of Lemma~\ref{lem:intermediatePotts} which as we saw in the previous section is the key ingredient to carry out the interpolation-scheme idea. We highlight here some of the key ideas (with a non-technical overview), which are also used to prove the analogous Lemma for obtaining our inapproximability results for 2-spin systems.

To prove Lemma~\ref{lem:intermediatePotts}, we will use three different types of gadgets.

The first type of gadgets, that have also been used in previous inapproximability results, are the so-called ``phase gadgets'', which are almost $\Delta$-regular bipartite graphs with a relatively small number of degree $\Delta-1$ vertices (the so-called ``ports''). This type of gadget exploits the phase transitions of the model and has $q$-ary behaviour, in the sense that a typical sample from their Gibbs distribution is in one of the $q$ ordered  phases, favoring one spin over the others. Aside from this $q$-ary behaviour, another feature of these gadgets is that they are convenient to maintain the degree  of the vertices in our constructions small, using the ports to make connections between gadgets.

The second type of gadgets are paths; these allow us to interpolate the edge activity $\beta$. The key point is that long paths induce some small edge-interaction $\beta$ between their endpoints (bigger than but close to $1$) and by using a big number of them (in parallel-style connections) we can achieve a target edge activity $\hat\beta$ with arbitrary good precision; here, the ports of the phase gadgets allow us to do these parallel connections without exceeding the degree bound $\Delta$. This is a crucial ingredient in implementing the new reduction idea.

The final type of gadgets consists of  the so-called edge-interaction gadgets. Each such gadget has two vertices, say $\rho,\rho'$, which we also refer to as ports. We are interested in two quantities of these  gadgets (cf. Definition~\ref{def:edgeinte33}):
\begin{itemize}
\item the effective edge activity, i.e., the relative ratio of the aggregate weight of configurations where $\sigma(u)=\sigma(v)$ versus $\sigma(u)\neq \sigma(v)$. Note that this ratio is always bigger than 1, due to the ferromagnetic interaction.
\item the susceptibility gap, i.e., the difference between the expected susceptibility conditioned on $\sigma(u)=\sigma(v)$ and the susceptibility conditioned on $\sigma(u)\neq \sigma(v)$.
\end{itemize}
We prove the existence of pairs of susceptibility gadgets which have roughly equal induced edge parameters but different susceptibility gaps. The equality between the induced edge parameters allows us to use them as probes (without changing the underlying distribution) for ``susceptibility'' between two vertices $s,t$, i.e., the probability that $s,t$ have the same colour, in a graph $G$. That is, we can invoke a presumed oracle for susceptibility when we use the first gadget (by identifying $s,t$ with the terminals) and get a ``reading'' for susceptibility, and do the same for the second and get a second ``reading''; the difference between the two readings gives us information about the probability that $s,t$ have the same colour in the original graph $G$. 

The reason that these susceptibility gadgets are useful is that analysing the susceptibility of the other two types of gadgets is deeply unpleasant and, in fact, it is not even known how to obtain susceptibility estimates for the phase gadgets (since their analysis in earlier works builds upon second moment methods that give rather crude bounds in our setting). Hence, by the subtraction trick above, we have the required modularity to avoid such refined considerations.

That said, establishing the existence of pairs of susceptibility gadgets with the required properties has various challenges and the proof is based on an elaborate construction which finishes by a contradiction argument via Cauchy's functional equation. Fortunately, this ground work has been largely done in \cite{averages}, though in our setting we need to consider edge gadgets instead of vertex gadgets, which complicates the underlying functions involved in the proofs. We believe that these constructions can be used to strengthen the results of \cite{averages} and obtain inapproximability for multi-spin systems such as colourings or the  antiferromagnetic Potts model.

These ideas suitably adapted apply to obtain our inapproximability results for antiferromagnetic 2-spin systems. The difference for 2-spin systems is that the interpolation is quite trickier, since in the setting there it is harder to make vertex or edge activities that are close to 1 and do the interpolation (in contrast to the paths used above which is the fairly natural choice). Instead, to do the interpolation, we use a pair of trees whose induced vertex activities (at the root) are sufficiently close and which are attached (in appropriate numbers) to the ports of the phase gadgets to imitate the effect of an external field close to 1. We are then able to interpolate in terms of $\lambda$ by a suitable implementation of the subtraction trick idea; we again need to depart from \cite{averages} (where 2-spin models were also considered) since the construction there does not yield a suitable interpolation parameter. The final new ingredient is to account for the general vertex-edge observables, since a key fact used in \cite{averages} is that the magnetization is an appropriate derivative of the log-partition function, which is no longer the case for general vertex-edge observables. 

We now  state more formally the above ingredients and show how to combine these and conclude the proof of Lemma~\ref{lem:intermediatePotts}.

\subsection{The gadgets}
\subsubsection{Bipartite phase gadgets for the Potts model}\label{sec:phasegadgets}
For integers $t,n,\Delta$, we let  $\Gc^t_{n,\Delta}$ be the distribution on bipartite graphs where there are $n$ vertices with degree $\Delta$ on each side, and $t$ vertices of degree $\Delta-1$ on each side. For a graph $G\in\Gc^t_{n,\Delta}$,  we denote the set of vertices with degree $\Delta$ by $U$  and by $W$ those with degree $\Delta-1$, so that $|U|=2n$ and $|W|=2t$. We will refer to set $W$ as the \textit{ports} of $G$.  For $\sigma: U\rightarrow [q]$, we define the \emph{phase} $\Yc(\sigma)$ of the configuration $\sigma$ as the most frequent color (breaking ties arbitrarily), i.e., which has the most occupied vertices under $\sigma$, i.e., 
\[\Yc(\sigma)=\arg\max_{i\in [q]}|\sigma^{-1}(i)|\]  
Let $p>1/q$ be given from $p=\tfrac{x}{x+q-1}$ where $x>1$ is the largest solution of $x=\big(\tfrac{\beta x+q-1}{x+\beta+q-2}\big)^{\Delta-1}$, cf. \cite[Footnote 5]{PottsBIS}. For a colour $i\in [q]$, we  define the product measure  $Q^{i}_{W}(\cdot)$ on configurations $\tau:W\rightarrow [q]$, given by
\[Q^{i}_W(\tau)=p^{|\tau^{-1}(i)|}\big(\tfrac{1-p}{q-1}\big)^{|W|-|\tau^{-1}(i)|}.\]
We will need the following two properties from the phase gadget $G$ for some sufficiently small $\epsilon>0$. Let $\mu:=\mu_{G;q,\beta}$.
\begin{enumerate}
	\item \label{it:Potts1}The $q$ phases appear with roughly equal probability, i.e., $|\mu(\Yc(\sigma)=i) - \tfrac{1}{q}\big| \leq \epsilon$ for $i\in[q]$.
  \item \label{it:Potts2}For $i\in [q]$ and any $\tau:W\rightarrow[q]$, $\big|
       \frac{\mu\big(\sigma_{W}=\tau\,\mid\, \Yc(\sigma)=i\big)}
       {Q^{i}_W(\tau)}
       -1
       \big|
       \leq  \epsilon$.
\end{enumerate}
Let $\Gc^{t,\epsilon}_{n,\Delta}$ denote the set of graphs $G\in \Gc^t_{n,\Delta}$ satisfying Items~\ref{it:Potts1} and~\ref{it:Potts2}. The following lemma is shown in \cite{PottsBIS}.
\begin{lemma}[{\cite[Lemma 28]{PottsBIS}}]
\label{lem:slypotts}
Let $q,\Delta\geq 3$ be integers and $\beta>\beta_c(q,\Delta)$. Then, there is a randomized algorithm that, on input integer $t\geq 1$ and $\epsilon>0$, outputs in time $\mathrm{poly}(t,\tfrac{1}{\epsilon})$ an integer $n$ and a graph $G$ that belongs to $\Gc^{t,\epsilon}_{n,\Delta}$,  with probability $\geq 3/4$.
\end{lemma}

\subsubsection{Edge-interaction/susceptibility gadgets}\label{sec:edgegadgets}
\begin{definition}\label{def:edge}
An {\em edge-interaction gadget} is a connected series-parallel graph $\Ec$ with two distinct vertices $\rho,\rho'$ that have degree one. We will refer to $\rho,\rho'$ as the {\em ports} of $\Ec$. 
\end{definition}

\begin{definition}\label{def:edgeinte33}
Let $\Ec$ be an edge-interaction gadget with ports $\rho,\rho'$, and $\mu=\mu_{\Ec;\beta}$. We denote by $B_{\Ec}=B_{\Ec}(\beta)$ the {\em effective interaction} of the gadget, i.e., $B_{\Ec}=\frac{\mu(\sigma_\rho=\sigma_{\rho'}=1)}{\mu(\sigma_\rho=1, \sigma_{\rho'}=2)}$ and by $S_{\Ec}=S_{\Ec}(\beta)$ the {\em susceptibility gap} of the gadget, i.e., $S_{\Ec}=\Eb_{\sigma\sim \mu}[\,\mc_{\Ec}(\sigma)\mid \sigma_\rho=\sigma_{\rho'}]-\Eb_{\sigma\sim\mu}[\,\mc_{\Ec}(\sigma)\mid \sigma_\rho\neq\sigma_{\rho'}]$.
\end{definition}

The following ``interaction'' gadget will allow us to change the edge interaction parameter to any desired value.
\newcommand{\statelemintergadget}{Let $q\geq 2$ be an integer and $\beta>1$ be a real. There is an algorithm, which, on input a rational $r\in(0,1/2)$, 
outputs in time  $poly(\bit(r))$ a path  $\Pc$ of length $O(|\log r|)$, such that $0<B_{\Pc} - 1< r$.}
\begin{lemma}\label{lem:intergadget}
\statelemintergadget
\end{lemma}
The proof of Lemma~\ref{lem:intergadget} is given in Section~\ref{sec:susc1Potts}. The following lemma gives pairs of edge-interaction gadgets which have almost the same edge interaction but different susceptibility gaps; this difference in the susceptibility gaps while maintaining the edge interaction will be the key to read off the susceptibility by subtraction.
\begin{lemma}\label{lem:suscgadget}
Let $q\geq 3$ be an integer and $\beta>1$ be a real. For any arbitrarily small constant $\delta>0$, there exist constants $S,\Xi>0$ and $B\in (1,1+\delta)$ such  that the following holds. There is an algorithm, which, on input a rational $r\in(0,1/2)$, outputs in time  $poly(\bit(r))$ a pair of edge-interaction gadgets  $\mathcal{E}_1,\mathcal{E}_2$, each of maximum degree $3$ and size $O(|\log r|)$, such that
\[|B_{\Ec_1} - B|,\, |B_{\Ec_2} - B|\leq r,\mbox{ but } |S_{\Ec_1}-S_{\Ec_2}|\geq S.\]
Moreover, the susceptibility gaps  $|S_{\Ec_1}|, |S_{\Ec_2}|$ are upper-bounded in absolute value by the constant $\Xi$, i.e., $|S_{\Ec_1}|, |S_{\Ec_2}|\leq \Xi$.
\end{lemma}
The proof of Lemma~\ref{lem:suscgadget} generalises the techniques from \cite{averages}, and is given in Section~\ref{sec:lastproofs}.

\subsection{The reduction -- proof of Lemma~\ref{lem:intermediatePotts}}\label{sec:rgtgtgrtte}
Let $q,\Delta\geq 3$ be integers and $\beta>\beta_c(q,\Delta)$. Let $H$ be a cubic bipartite graph which is input to the problem $\Susc(q)$ of Section~\ref{sec:proofthm}. For integers $n,t\geq 1$ and rational $\epsilon>0$, let $G\in \calG^{t,\epsilon}_{n,\Delta}$ be a bipartite phase gadget satisfying Items~\ref{it:Potts1} and~\ref{it:Potts2} of Section~\ref{sec:phasegadgets}. Let $\Ec$ be an edge-interaction gadget with effective interaction $B_\Ec$ and susceptibility gap $S=S_{\Ec}$. Let $\Pc$ be a path with effective edge interaction $B_\Pc$. 

For an integer $\ell$ satisfying $\ell<t/3$, we define the graph $H^\ell_{G,\Ec, \Pc}$ as follows. For each vertex $v$ of $H$ replace it with a distinct copy of $G$, denoted by $G_v$; we also use $U_v,W_v$ to denote the sets corresponding to $U,W$ in $G_v$. Moreover for each $\{u,v\}$ of $H$, add a matching of size $\ell+1$ between $W_u$ and $W_v$,  and replace $\ell$ edges of the matching by (distinct) copies of the path $\Pc$ and the last edge of the matching by the gadget $\Ec$.  Since $H$ is bipartite, this constuction can clearly be done so that the final graph $H^\ell_{G,\Ec, \Pc}$ obtained this way is bipartite.
Let $H^{\ell}_{G,\Pc}$ be the graph with the copies of the susceptibility gadget removed.

The lemma below relates the susceptibility $\Sc_{q,\beta}(H^\ell_{G,\Ec,\Pc})$ with the susceptibility of $\Sc_{q,\hat\beta}(H)$, for some appropriate $\hat\beta$ that is a function of the parameters $q,\Delta,\beta$ and $\ell, B_\Ec,B_\Pc$; we expain how the lemma corresponds to the overview of Section~\ref{sec:overview} right after its statement. The following piece of notation will be useful: for a graph $J$ and a subgraph $J'$ of $J$, given  a configuration $\sigma:V(J)\rightarrow [q]$, it will be convenient to denote by $m_{J'}(\sigma)=\sum_{e=\{u,v\}\in E(J')}\mathbf{1}\{\sigma(u)=\sigma(v)\}$ the number of monochromatic edges of $J'$ under $\sigma$.  
\newcommand{\statelemPottsconstruction}{Let $q,\Delta\geq 3$ be integers and $\beta>\beta_c(q,\Delta)$. There are constants $1>R_0>R_1>0$
so that the following holds for any path $\Pc$ with edge interaction $B_\Pc$, any edge-interaction gadget $\Ec$ with effective interaction $B_{\Ec}$ and susceptibility gap $S_{\Ec}$, and any integers $\ell,t$ with $t\geq 3(\ell+1)$.

For a cubic bipartite graph $H$, for any $\epsilon\leq \tfrac{1}{(5q|V(H)|)^{2}}$, any integer $n$ and phase gadget $G\in \Gc^{t,\epsilon}_{n,\Delta}$, for $\mu:=\mu_{H^{\ell}_{G,\Ec,\Pc}}$ and $\epsilon'=10q |V(H)| \epsilon$, it holds that
\[\Sc_{q,\beta}(H^\ell_{G,\Ec,\Pc})=\Ac_\Ec|E(H)|+\Eb_{\sigma\sim \mu}[m_{H^\ell_{G,\Pc}}(\sigma)]+(1\pm \epsilon') S_\Ec \big[(A_{0}-A_{1}) \Sc_{q,\hat\beta}(H)+A_{\nequal}|E(H)|\big],\]
where  $A_{\Ec}= \Eb_{\sigma\sim \mu_{\Ec}}[m_{\Ec}(\sigma)\mid \sigma_\rho=\sigma_{\rho'}]$ and 
\[\hat\beta:=\big(\tfrac{1+(B_{\Pc}-1)R_{\equal}}{1+(B_{\Pc}-1)R_{\nequal}}\big)^{\ell}\big(\tfrac{1+(B_{\Ec}-1)R_{\equal}}{1+(B_{\Ec}-1)R_{\nequal}}\big),\qquad A_{j}:=\tfrac{B_{\Ec}}{B_{\Ec}+\tfrac{1-R_j}{R_j}}\mbox{ for }j\in \{0,1\}.\]}
\begin{lemma}\label{lem:Pottsconstruction}
\statelemPottsconstruction
\end{lemma}
To give a bit of intuition behind the expression of $\Sc_{q,\beta}(H^\ell_{G,\Ec,\Pc})$, recall that the vertices of $H$ are replaced with copies of the bipartite phase gadgets $G$ and that for each pair of neighboring vertices in $H$ we connect the corresponding copies of $G$ using the appropriate number of the gadgets $\Ec,\Pc$. The point here is that the bipartite gadgets are so large that each one of them is with very high probability in one of the $q$ phases (cf. Item~\ref{it:Potts1} in Section~\ref{sec:phasegadgets}) and therefore in the Gibbs distribution of $H^\ell_{G,\Ec,\Pc}$ (with parameters $q,\beta$) they behave like meta-vertices which are in one of $q$ states, analogously to a Potts model on $H$ with $q$ spins and a new edge activity $\hat\beta$, which is ultimately determined by the $\Ec,\Pc$-connections and the (induced) probability distributions on the ports of the bipartite phase gadgets (conditioned on the phase, cf. Item~\ref{it:Potts2} in Section~\ref{sec:phasegadgets}). This explains (at an intuitive level) the presence of the quantity $\Sc_{q,\hat\beta}(H)$; the remaining terms are offsets to account for the addition of the various gadgets. Of those, the most complicated is the term $\Eb_{\sigma\sim \mu}[m_{H^\ell_{G,\Pc}}(\sigma)]$ which involves the contribution to the susceptibility from edges in the graph $\Eb_{\sigma\sim \mu}[m_{H^\ell_{G,\Pc}}(\sigma)]$ which is hard to get a neat expression since the average is taken over the complicated distribution $\mu$. This is where the idea of having a pair of susceptibility gadgets $(\Ec_1,\Ec_2)$ with the same effective interaction but different susceptibility gaps will come into play (in the proof of Lemma~\ref{lem:intermediatePotts} below): by subtracting the susceptibilities for the graphs $H^\ell_{G,\Ec_1,\Pc}$ and $H^\ell_{G,\Ec_2,\Pc}$ between these, the terms corresponding to $\Eb_{\sigma\sim \mu}[m_{H^\ell_{G,\Pc}}(\sigma)]$ will cancel (since $\Ec_1,\Ec_2$ have roughly the same effective interaction $B_{\Ec_1},B_{\Ec_2}$) allowing us to approximate the target quantity $\Sc_{q,\hat\beta}(H)$ (since $\Ec_1,\Ec_2$ have substantially different susceptibility gaps $S_{\Ec_1},S_{\Ec_2}$). That said, the proof of Lemma~\ref{lem:Pottsconstruction} is on the technical side and is deferred to Section~\ref{sec:deferred} of the full version.

To finish the proof of Lemma~\ref{lem:intermediatePotts}, we need the following crude bound on the change of susceptibility when we slightly change the values of the edge activities on a subset of the edges. To state the lemma, for a graph $G$ with edge-activity vector $\betab=\{\beta_e\}_{e\in E(H)}$, define the weight of an assignment $\sigma:V(G)\rightarrow[q]$ as $w(\sigma)=\prod_{e=\{u,v\}\in E(G)}(\beta_e)^{\mathbf{1}\{\sigma(u)=\sigma(v)\}}$, and let $\mu_{G;q,\betab}(\sigma)=w(\sigma)/Z_{G;q,\betab}$ denote the corresponding Gibbs distribution, where $Z_{G;q,\betab}$ is the normalizing constant.
\begin{lemma}\label{lem:t45t45perturb}
Let $H$ be a graph and $F$ be a subgraph of $H$ on the same set of vertices. Suppose that $\betab=\{\beta_e\}_{e\in E(H)},\betab'=\{\beta_e'\}_{e\in E(H)}$ are edge activity vectors  such that $\beta_e=\beta_e'=\beta$ for $e\in E(F)$, and $\beta_e=\beta_0$, $\beta_e'=\beta_1$ for $e\notin E(F)$. Then, for $\mu:=\mu_{H;q,\betab}$ and $\mu':=\mu_{H;q,\betab}'$, it holds that 
\[\Big|\Eb_{\sigma\sim \mu}[m_{F}(\sigma)]-\Eb_{\sigma\sim \mu'}[m_{F}(\sigma)]\Big|\leq |E(H)|^2 |\beta_0-\beta_1|.\]
\end{lemma}
\begin{proof}
Suppose without loss of generality that $\beta_0\geq \beta_1$; by the monotonicity of the ferromagnetic Potts model we have that $\Eb_{\sigma\sim \mu}[m_{F}(\sigma)]\geq \Eb_{\sigma\sim \mu'}[m_{F}(\sigma)]$ (see, e.g., \cite[Theorems 1.16 \& 3.21]{grimmett2004random}).  For a configuration $\sigma: V(H)\rightarrow[q]$, let $w(\sigma),w'(\sigma)$ denote its weight under the edge activity vectors $\betab$ and $\betab'$, respectively. Consider an edge $e\in E(F)$. Then, using that for reals $a>b>0$ it holds that  $|a^k-b^{k}|\leq k|a-b|a^{k}$, we obtain that for every $\sigma$ it holds that 
\[0<w(\sigma)-w'(\sigma)=\beta^{m_{F}(\sigma)}\big(\beta_0^{|E(H)|-m_F(\sigma)} -(\beta_1)^{|E(H)|-m_F(\sigma)}\big)\leq |E(H)|(\beta_0-\beta_1) w(\sigma)\]
By summing over $\sigma$, it follows also that $Z_{G;q,\betab'}\leq Z_{G;q,\betab}$, and combining these we obtain that
\[\Eb_{\sigma\sim \mu}[m_{F}(\sigma)]-\Eb_{\sigma\sim \mu'}[m_{F}(\sigma)]\leq \frac{\sum_{\sigma}m_{F}(\sigma)\big(w(\sigma)-w'(\sigma)\big)}{Z_{G;q,\betab}}\leq |E(H)|^2(\beta_0-\beta_1).\qedhere\]
\end{proof}
We now give the proof of Lemma~\ref{lem:intermediatePotts} which we restate here.
\begin{lemintermediatePotts}
\statelemintermediatePotts
\end{lemintermediatePotts}
\begin{proof}
Let $H$ be a cubic bipartite graph and $\hat\beta>1$ be the inputs to $\SuscCubic(q)$, and let $\eta\in (0,1)$ be the desired relative error that we want to approximate $\Sc_{q,\hat\beta}(H)$. We may assume that $\hat\beta\geq \beta_0= \big(\tfrac{q-1}{3}\big)^{1/\Delta}$;  for  $\hat\beta<\beta_0$,  a fairly standard coupling argument shows that Glauber dynamics converges rapidly to the Gibbs distribution $\mu_{H;q,\hat\beta}$, see for example \cite[Theorem 1.1]{mixingPotts}, and therefore it can be used to approximate $\Sc_{q,\hat\beta}(H)$ in time $poly\big(V(H),\tfrac{1}{\eta},\bit(\hat{\beta})\big)$ using rejection sampling. For some of the bounds below, it will also be convenient to assume that $|V(H)|,|E(H)|$ are bigger than a sufficiently large constant (otherwise, we can just brute-force).

Let $1>R_{\equal}>R_{\nequal}>0$ be the constants in Lemma~\ref{lem:Pottsconstruction}, and let $\delta\in (0,1)$ be a rational constant such that for all $B\in(1,1+\delta)$, it holds that $\tfrac{1+(B-1)R_{\equal}}{1+(B-1)R_{\nequal}}\leq \beta_0<\hat\beta$. Note that the choice of $\delta$ is a constant depending on $q,\Delta$ but  independent of $H$ and $\hat\beta$.   By Lemma~\ref{lem:suscgadget}, there are constants $B\in (1,1+\delta)$, $S>0$ and an algorithm, which, on input a rational $r\in(0,1/2)$, outputs in time  $poly(\bit(r))$ a pair of susceptibility gadgets  $\mathcal{E}_1,\mathcal{E}_2$, each of maximum degree $3$ and size $O(|\log r|)$, such that
\begin{equation}\label{eq:BguarP}
|B_{\Ec_1} - B|,\, |B_{\Ec_2} - B|\leq r,\mbox{ but } |S_{\Ec_1}-S_{\Ec_2}|\geq S.
\end{equation}

Let $\epsilon=\tfrac{\eta}{|E(H)|^5}$ and $t=\big\lceil \big(\tfrac{|E(H)|\log \hat\beta}{\epsilon\delta(B-1)}\big)^4\big\rceil$. By Lemma~\ref{lem:slypotts}, there is an algorithm that outputs in time $poly(t,\tfrac{1}{\epsilon})$ an integer $n$ and a graph $G\in {\cal G}^{t,\epsilon}_{n,\Delta}$ (satisfying Items~\ref{it:Potts1} and~\ref{it:Potts2}).
%Moreover, the susceptibility gaps  $|S_{\Ec_1}|, |S_{\Ec_2}|$ are upper-bounded in absolute value by the constant $\Xi$, i.e., $|S_{\Ec_1}|, |S_{\Ec_2}|\leq \Xi$.
Use the algorithm of Lemma~\ref{lem:suscgadget} to obtain  gadgets $\Ec_1,\Ec_2$ satisfying \eqref{eq:BguarP} for $r=\tfrac{\epsilon^4}{10\delta(R_0-R_1)\beta_0}$. Moreover, use Lemma~\ref{lem:intergadget}, to obtain in time $poly(\bit(r))$ an edge interaction gadget with $1<B_{\Pc}<1+\epsilon$. Let $\ell$ be the smallest positive integer such that
\[\big(\tfrac{1+(B_\Pc-1)R_{\equal}}{1+(B_\Pc-1)R_{\nequal}}\big)^{\ell}\big(\tfrac{1+(B-1)R_{\equal}}{1+(B-1)R_{\nequal}}\big)>\hat\beta\]
and note that such an integer exists by the choice of $\delta$ since the l.h.s. for $\ell=0$ is smaller than $\hat\beta$, and each of the fractions is bigger than 1 from $R_0>R_1$ and $B>1$. In fact, we have that $\ell=O(\tfrac{1}{\epsilon}\log \hat\beta)$, where the implicit constants depend only on $q,\Delta$. It follows in particular that $\ell<t/3$.

For $i\in \{1,2\}$, consider now the graphs $\widehat{H}_i=H^{\ell}_{G,\Pc,\Ec_i}$ and let $\mu_i=\mu_{\widehat{H}_i;q,\beta}$. From Lemma~\ref{lem:Pottsconstruction}, we have that
\[\Sc_{q,\beta}(\widehat{H}_i)=\Ac_{\Ec_i}|E(H)|+\Eb_{\sigma\sim \mu_i}[m_{H^\ell_{G,\Pc}}(\sigma)]+(1\pm \eta^2) S_{\Ec_i} \Big[\big(A^{(i)}_{0}-A^{(i)}_{1}\big) \Sc_{q,\hat\beta_i}(H)+A^{(i)}_{\nequal}|E(H)|\Big],\]
where  $A_{\Ec_i}= \Eb_{\sigma\sim \mu_{\Ec_i}}[m_{\Ec_i}(\sigma)\mid \sigma_\rho=\sigma_{\rho'}]$ and 
\[\hat\beta_i:=\big(\tfrac{1+(B_{\Pc}-1)R_{\equal}}{1+(B_{\Pc}-1)R_{\nequal}}\big)^{\ell}\big(\tfrac{1+(B_{\Ec_i}-1)R_{\equal}}{1+(B_{\Ec_i}-1)R_{\nequal}}\big),\qquad A^{(i)}_{j}:=\tfrac{B_{\Ec_i}}{B_{\Ec_i}+\tfrac{1-R_j}{R_j}}\mbox{ for }j\in \{0,1\}.\] 
From \eqref{eq:BguarP}, we have that $\hat\beta_i=(1\pm \epsilon^3)\hat\beta$, and therefore from Lemma~\ref{lem:t45t45perturb}, we have that
\begin{gather*}
\big|\Eb_{\sigma\sim \mu_1}[m_{H^\ell_{G,\Pc}}(\sigma)]-\Eb_{\sigma\sim \mu_2}[m_{H^\ell_{G,\Pc}}(\sigma)]\big|\leq |E(H^\ell_{G,\Ec,\Pc})|\epsilon^3\leq \epsilon^2,\\
|\Sc_{q,\hat\beta_i}(H)-\Sc_{q,\hat\beta}(H)|\leq |E(H)|^2\epsilon^3\leq \epsilon^2
\end{gather*}
Using \eqref{eq:BguarP}, we also have that 
\[A^{(i)}_0=(1\pm \epsilon)A_0,A^{(i)}_1=(1\pm \epsilon^2)A_i, \mbox{ where } A_{j}:=\tfrac{B}{B+\tfrac{1-R_j}{R_j}}\mbox{ for }j\in \{0,1\}.\]
We can invoke the oracle for $\Sc_{q,\beta}(\widehat{H}_i)$ to compute $\hat{\Sc}_i$ such that $\hat{\Sc}_i=(1\pm \epsilon^2)\Sc_{q,\beta}(\widehat{H}_i)$. Note also that $\Ec_i$ has size $poly(\bit(r))$ and therefore we can invoke the oracle for $\Susc(q,\Delta,\beta)$ to compute $\hat{A}_{\Ec_i}, \hat{S}_{\Ec_1}$ such that $\hat{A}_{\Ec_i}=(1\pm \epsilon^2)A_{\Ec_i}$ and $\hat{S}_{\Ec_i}=(1\pm \epsilon^2)S_{\Ec_i}$.\footnote{For $\hat{A}_{\Ec_i}$, just invoke the oracle on the graph obtained from $\Ec_i$ by identifying $\rho_i$ and $\rho_i'$; this graph has maximum degree at most $3$ since $\rho_i,\rho_i'$ both have degree 1. Observe that $S_{\Ec_i}=2A_{\Ec_i}-\Sc_{q,\beta}(\Ec_i)$ and therefore we can obtain the desired $\hat{S}_{\Ec_i}$ by using a further oracle call on $\Ec_i$ to approximate $\Sc_{q,\beta}(\Ec_i)$.}
It follows that
\[\hat{S}=\frac{1}{A_0-A_1}\Big(\frac{(\hat{\Sc}_1-\hat{\Sc}_2)-|E(H)|(\hat{A}_{\Ec_1}-\hat{A}_{\Ec_2})}{\hat{S}_{\Ec_1}-\hat{S}_{\Ec_2}}-A_1 |E(H)|\Big)\]
is within a factor of $(1\pm \eta)$ of the susceptibility $\Sc_{q,\hat\beta}(H)$, as needed. This finishes the reduction and therefore the proof of Lemma~\ref{lem:intermediatePotts}. 
\end{proof}

\section{Hardness of vertex-edge observables for 2-spin systems}\label{sec:2spin}
Throughout this section, we will fix integer $\Delta\geq 3$, and antiferromagnetic $(\beta,\gamma,\lambda)\in \Uc_\Delta$ in the non-uniqueness region. We will also fix an $(a,b,c)$ vertex-edge observable that is non-trivial on bipartite graphs. For a graph $G=(V,E)$ and an assignment $\sigma:V'\rightarrow\two$ where $V'$ is a superset of $V$, it will be convenient to denote by $\oc_G(\sigma)=a|\sigma_V|+b m_0(\sigma_V)+c m_1(\sigma_V)$ the value of the observable on the graph $G$ with respect to the restriction of the assignment of $\sigma$ onto $V$.

\subsection{The interpolation scheme}\label{sec:43frfe4t}
Analogously to Section~\ref{sec:proofthm}, it will be convenient to consider the following computational problems. 
\prob{$\Obser(\beta,\gamma,\lambda,a,b,c)$.} 
{ A bipartite graph $G$ with max degree $\Delta$.} 
 { The $(a,b,c)$ vertex-edge observable on $G$ with parameters $\beta$,$\gamma$,$\lambda$, i.e., the value $\Oc_{\beta,\gamma,\lambda}(G)$.}

In the second, the parameter is going to be the edge activity $\alpha<1$ of an antiferromagnetic Ising model; note that the problem allows the vertex activity to be part of the input.
\prob{$\MagIsing(\alpha)$.} 
{ A cubic bipartite graph $H$, and a rational vertex activity $\hat\lambda>0$.} 
 { The magnetization on $H$ for the Ising model with parameters $\alpha,\hat\lambda$, i.e., the value $\Mc_{\alpha,\alpha,\hat\lambda}(H)$.}

We now show the following analogue of the interpolation scheme in Lemma~\ref{lem:intermediatePotts}.
\newcommand{\statelemintermediateIsing}{Let $\Delta\geq 3$ be an integer and $(\beta,\gamma,\lambda)\in \Uc_{\Delta}$. Then, there is $\alpha\in (0,1)$ such that for any $(a,b,c)$ vertex-edge  observable that is not trivial on bipartite graphs,  
\[\MagIsing(\alpha)\leq_{\mathrm{AP}} \Obser(\beta,\gamma,\lambda,a,b,c).\]}
\begin{lemma}\label{lem:intermediateIsing}
\statelemintermediateIsing
\end{lemma}
Assuming the key Lemma~\ref{lem:intermediateIsing}, the proof of Theorem~\ref{thm:main2} can be done analogously to Theorem~\ref{thm:main1}, interpolating now in terms of the vertex activity $\hat\lambda$.  We defer the proof to Section~\ref{sec:proof-2spin} of the Appendix.

\subsection{The gadgets}
In this section, we outline the gadgets that will be used to prove Lemma~\ref{lem:intermediateIsing}. These are analogous to those presented in the case of the Potts model, especially the phase gadgets. To account for general vertex-edge observables, we refine appropriately the field-gadget idea of \cite{averages}, by now paying attention to the so-called observable gap (cf. Definition~\ref{def:gapobservable}).
\subsubsection{Bipartite phase gadgets for antiferromagnetic 2-spin systems}\label{sec:phase2spingadgets}
We follow the same notation as in Section~\ref{sec:phasegadgets} to denote for integers $t,n,\Delta$ the class  $\Gc^t_{n,\Delta}$ of bipartite graphs where there are $n$ vertices with degree $\Delta$ on each side, and $t$ vertices of degree $\Delta-1$ on each side. For a graph $G\in\Gc^t_{n,\Delta}$,  we denote its bipartition by $(U^\pl,U^\mi)$ where $U^\pl,U^\mi$ are vertex sets with $|U^\pl|=|U^\mi|=n$, and we denote by $W^\pl,W^\mi$ the sets of vertices with degree $\Delta-1$ on each side of the bipartition,  so that $|W^\pl|=|W^\mi|=t$. We will refer to set $W=W^+\cup W^-$ as the \textit{ports} of $G$.  For $\sigma: U\rightarrow \two$, we define the \emph{phase} $\Yc(\sigma)$ of the configuration $\sigma$ as the side of the bipartite graph which has the most occupied vertices under $\sigma$, i.e., 
\[\Yc(\sigma)=\arg\max_{i\in \{\pl,\mi\}}|\sigma^{-1}(1)\cap U^{i}|.\] 
It is known that for $(\beta,\gamma,\lambda)\in \Uc_{\Delta}$ the system of equations  $x=\frac{1}{\lambda}\left(\frac{ \beta y+1}{y+\gamma}\right)^{\Delta-1},y=\frac{1}{\lambda}\left(\frac{ \beta x+1}{x+\gamma}\right)^{\Delta-1}$ has a unique solution with $y>x>0$, see, e.g., \cite[Lemma 7]{GSVIsing}. Let $q_{\pl}=\tfrac{1}{1+x}$, $q_\mi=\tfrac{1}{1+y}$ and note that $q_\pl,q_\mi$ are distinct numbers in the interval $(0,1)$. Define the product distributions $Q^\pl_W(\cdot), Q^\mi_W(\cdot)$  by 
\begin{equation}\label{eq:2spinproduct}
Q^\plm_W(\tau)= (q_\plm)^{|\tau^{-1}(1) \cap W^\pl|}(1-q_\plm)^{|\tau^{-1}(0) \cap W^\pl|}(q_\mip)^{|\tau^{-1}(1) \cap W^\mi|}(1-q_\mip)^{|\tau^{-1}(0) \cap W^\mi|}.
\end{equation} 

We will need the following two properties from the phase gadget $G$ for some sufficiently small $\epsilon>0$. Let $\mu:=\mu_{G;\beta,\gamma,\lambda}$.
\begin{enumerate}
	\item \label{it:2spin1}The phases appear with roughly equal probability, i.e., $|\mu(\Yc(\sigma)=\plm) - \tfrac{1}{2}\big| \leq \epsilon$.
  \item \label{it:2spin2}For any $\tau:W\rightarrow\two$, $\Big|
       \displaystyle \frac{\mu\big(\sigma_{W}=\tau\,\mid\, \Yc(\sigma)=\plm\big)}
       {Q^{\plm}_W(\tau)}
       -1
       \Big|
       \leq  \epsilon$.
\end{enumerate}
Let $\Gc^{t,\epsilon}_{n,\Delta}$ denote the set of graphs $G\in \Gc^t_{n,\Delta}$ satisfying Items~\ref{it:2spin1} and~\ref{it:2spin2}. The following lemma is is shown in \cite{CAI}.
\begin{lemma}[{\cite[Lemma 9]{CAI}}]
\label{lem:Cai2spin}
Let $\Delta\geq 3$ and $(\beta,\gamma,\lambda)\in \Uc_\Delta$.  Then, there is a randomized algorithm that, on input integer $t\geq 1$ and $\epsilon>0$, outputs in time $\mathrm{poly}(t,\tfrac{1}{\epsilon})$ an integer $n$ and a graph $G$ that belongs to $\Gc^{t,\epsilon}_{n,\Delta}$,  with probability $\geq 3/4$.
\end{lemma}

\subsubsection{Field gadgets with observable gaps}
We adopt the following definition of ``field'' gadgets from \cite{averages}.
\begin{definition}\label{def:field}
For $\lambda\neq \frac{1-\beta}{1-\gamma}$, a {\em field gadget} is a rooted tree $\Tc$ whose root $\rho$ has degree one. When $\lambda=\frac{1-\beta}{1-\gamma}$, a field gadget consists of a rooted bipartite graph obtained from a rooted tree where some  of the leaves have been replaced by a distinct cycle of length four (by identifying the leaf with a vertex of the cycle).  
\end{definition}

\begin{definition}\label{def:gapobservable}
Let $\Tc$ be a field gadget rooted at $\rho$, and $\mu=\mu_{\Tc;\beta,\gamma,\lambda}$. We denote by $R_{\mathcal{T}}=R_{\mathcal{T}}(\beta,\gamma,\lambda)$ the {\em effective field} of the gadget, i.e., $R_{\Tc}=\tfrac{1}{\lambda}\frac{\mu(\sigma_\rho=1)}{\mu(\sigma_\rho=0)}$ and by $O_{\Tc}=O_{\Tc}(\beta,\gamma,\lambda)$ the {\em observable gap} of the gadget, i.e., $O_{\Tc}=\Eb_{\sigma\sim \mu}[\,\oc_{\Tc}(\sigma)\mid \sigma_\rho=1]-a-\Eb_{\sigma\sim\mu}[\,\oc_{\Tc}(\sigma)\mid \sigma_\rho=0]$.
\end{definition}
The division by $\lambda$ in the definition of the effective field of a gadget is to avoid double-counting the contribution of the root later on.

\begin{lemma}\label{lem:tcplus}
Let $(\beta,\gamma,\lambda)$ be antiferromagnetic such that at least one of $\beta\neq \gamma$ or $\lambda\neq 1$ holds.  There are constants $C,\tilde R>0$ with $\tilde R\neq 1$ and an algorithm which, on input a rational $r\in(0,1/2)$, outputs in time  $poly(\bit(r))$  field gadgets  $\Tc_\pl,\Tc_\mi$, each of maximum degree $3$ and size $O(|\log r|)$, such that
\[R_{\Tc_\pl}> R_{\Tc_\mi}+r/2 \mbox{ and } |R_{\Tc_\pl} - \tilde R|,\ |R_{\Tc_\mi} - \tilde R|\leq r.\]
\end{lemma}

\begin{theorem}\label{thm:observablegadget} 
Let $(\beta,\gamma,\lambda)$ be antiferromagnetic, and consider any non-trivial vertex-edge observable on general graphs. There exist constants $\hat{R},\hat{O},\Xi>0$ and  an algorithm, which, on input a rational $r\in(0,1/2)$, outputs in time  $poly(\bit(r))$ a pair of field gadgets  $\Tc_1,\Tc_2$, each of maximum degree $3$ and size $O(|\log r|)$, such that
\[|R_{\Tc_1} - \hat{R}|,\, |R_{\Tc_2} - \hat{R}|\leq r,\mbox{ but } |O_{\Tc_1}-O_{\Tc_2}|\geq \hat{O}.\]
Moreover, the observable gaps  $O_{\Tc_1}, O_{\Tc_2}$ are upper-bounded in absolute value by the constant~$\Xi$.
\end{theorem}
The proofs of Lemma~\ref{lem:tcplus} and Theorem~\ref{thm:observablegadget} follow closely the approach in \cite{averages}, and are therefore deferred to  Section~\ref{sec:lastproofs}.

\subsection{The reduction}
Let $H$ be a cubic bipartite graph which is input to the problem $\MagIsing(\alpha)$ of Section~\ref{sec:43frfe4t}, for some constant $\alpha\in(0,1)$ to be specified.  For integers $n,t\geq 1$ and rational $\epsilon>0$, let $G\in \calG^{t,\epsilon}_{n,\Delta}$ be a bipartite phase gadget satisfying Items~\ref{it:2spin1} and~\ref{it:2spin2} of Section~\ref{sec:phase2spingadgets}. Let $\Tc_\pl,\Tc_\mi,\Tc$ be field gadgets. Note that the gadgets $\Tc_\pl,\Tc_\mi$ serve a different role to that of $\Tc$, and in particular they will be used to interpolate over the vertex activity $\hat\lambda$.

To achieve this, for integers $\ell_\pl,\ell_\mi$ satisfying $t\geq 5+\max\{\ell_\pl,\ell_\mi\}$, we define the graph $H^{\ell_\pl,\ell_\mi}_{G,\Tc_\pl,\Tc_\mi, \Tc}$ as follows. For each vertex $v$ of $H$ replace it with a distinct copy of $G$, denoted by $G_v$; we denote by $U^{\plm}_v,W_v^{\plm}$ the sets corresponding to $U^{\plm},W^{\plm}$ in $G_v$. Moreover, for each $v\in V(H)$, attach one copy of the gadget $\Tc$ and $\ell_\pl$ copies of the gadget $\Tc_\pl$ on mutually distinct vertices of $W^{\pl}$ by identifying them with the corresponding roots. Similarly, attach $\ell_\mi$ copies of the gadget $\Tc_\mi$ on mutually distinct vertices of $W^{\mi}$. Let $\Tc_v$ be the copy of $\Tc_v$ corresponding to $v$, and $w_v$ be its root. Let $W_{\Tc}=\{w_v \mid v\in V(H)\}$ be the  set of all these roots. Further, for each edge $\{u,v\}$ of $H$, add an edge between $W_u^\pl$ and $W_v^{\pl}$, and an edge   between $W_u^\mi$ and $W_v^{\mi}$. 

Let $H^{\ell_\pl,\ell_\mi}_{G,\Tc_\pl,\Tc_\mi}$ be the graph without the internal vertices and edges of the copies of gadget $\Tc$, i.e., we keep only the roots $W_{\Tc}$ of the gadgets  in $H^{\ell_\pl,\ell_\mi}_{G,\Tc_\pl,\Tc_\mi}$. The following piece of notation will be useful: for a graph $J$ and a subgraph $J'$ of $J$, given  a configuration $\sigma:V(J)\rightarrow [q]$, it will be convenient to denote by $m_{J'}(\sigma)=\sum_{e=\{u,v\}\in E(J')}\mathbf{1}\{\sigma(u)=\sigma(v)\}$ the number of monochromatic edges of $J'$ under $\sigma$.

The following lemma relates the value of the observable $\Oc_{\beta,\gamma,\lambda}(H^{\ell_\pl,\ell_\mi}_{G,\Tc_\pl,\Tc_\mi, \Tc})$ with the magnetization $\Sc_{\alpha,\hat\lambda}(H)$, for some appropriate $\hat\lambda$ that is a function of the parameters $\beta,\gamma,\lambda$ and $\ell_\pl,\ell_\mi, R_{\Tc_\pl},R_{\Tc_\pl}, R_{\Tc}$. Analogously to Section~\ref{sec:rgtgtgrtte}, for a graph $J$ and a subgraph $J'$ of $J$, given  a configuration $\sigma:V(J)\rightarrow [q]$, it will be convenient to denote by $o_{J'}(\sigma)=a|\sigma_{V(J')}|+b m_0(\sigma_{V(J')})+cm_1(\sigma_{V(J')})$ be the contribution of $J'$ to the value of the observable on $J$.   
\newcommand{\statelemspinconstruction}{Let $\Delta\geq 3$ be an integer, $(\beta,\gamma,\lambda)\in \Uc_{\Delta}$, and $(a,b,c)$ be a vertex-edge observable.  Then, there are constants $q_\pl,q_\mi\in (0,1)$ with $q_\pl>q_\mi$ and $\alpha\in (0,1)$ so that the following holds for any field gadgets $\Tc_\pl,\Tc_\mi,\Tc$ with effective fields $ R_\pl,\hat R_\mi,R$ and observable gaps $O_1, O_2,O$, and any positive integers $\ell_\pl,\ell_\mi,t$ with $t\geq 5+\max\{\ell_\pl,\ell_\mi\}$.

For a cubic bipartite graph $H$, for any $\epsilon\leq \tfrac{1}{(5|V(H)|)^{2}}$, any integer $n$ and phase gadget $G\in \Gc^{t,\epsilon}_{n,\Delta}$, for $\mu:=\mu_{H^{\ell_\pl,\ell_\mi}_{G,\Tc_\pl,\Tc_\mi,\Tc}}$ and $\epsilon'=10 |V(H)| \epsilon$, it holds that
\begin{align*}
\Oc_{\beta,\gamma,\lambda}(H^{\ell_\pl,\ell_\mi}_{G,\Tc_\pl,\Tc_\mi,\Tc})
&=\Ac |V(H)|+\Eb_{\sigma\sim \mu}\big[\oc_{H^{\ell_\pl,\ell_\mi}_{G,\Tc_\pl,\Tc_\mi}}(\sigma)\big]
%\\
%&\hspace{4cm}
+(1\pm \epsilon') O \big[(q_\pl-q_\mi) \Mc_{\alpha,\hat\lambda}(H)+q_\mi|V(H)|\big],
\end{align*}
where  $\Ac= \Eb_{\sigma\sim \mu_{\Tc}}[\oc_\Tc(\sigma)\mid \sigma_\rho=0]$ and $\hat\lambda:=\big(\tfrac{q_\pl R+1-q_\pl}{q_\mi R+1-q_\mi}\big)\big(\tfrac{q_\pl R_{\pl}+1-q_\pl}{q_\mi R_{\pl}+1-q_\mi}\big)^{\ell_\pl}/\big(\tfrac{q_\pl R_{\mi}+1-q_\pl}{q_\mi R_{\mi}+1-q_\mi}\big)^{\ell_\mi}$.}
\begin{lemma}\label{lem:2spinconstruction}
\statelemspinconstruction
\end{lemma}
The proof of Lemma~\ref{lem:2spinconstruction} builds upon similar ideas to that of Lemma~\ref{lem:Pottsconstruction} (see in particular the discussion around there for how this blends with the overview of Section~\ref{sec:overview}) and is deferred to Section~\ref{sec:deferred} of the full version.

We will need the following bound on the change of the observable value when we change the vertex activities of a subset of the vertices. Namely, let $G=(V,E)$ be a graph and $(\beta,\gamma)$ be antiferromagnetic. For a vertex-activity vector $\lambdab=\{\lambda_v\}_{v\in V}$, define the Gibbs distribution $\mu_{G;\beta,\gamma,\lambdab}(\sigma)\propto\beta^{m_{0}(\sigma)}\gamma^{m_{1}(\sigma)}\prod_{v\in V;\, \sigma(v)=1} \lambda_v$ for $\sigma:V\rightarrow \two$. 

\begin{lemma}[Minor adaptation of {\cite[Lemma 35]{averages}}]\label{lem:2spinperturb}
Let $(\beta,\gamma)$ be antiferromagnetic,  $\lambda,\lambda_1,\lambda_2>0$, and $(a,b,c)$ be a vertex-edge observable. Let $G=(V,E)$ be a graph and $S\subseteq V$. For $i\in \{1,2\}$, let $\lambdab_i$ be the field vector on $V$, where every $v\in S$ has activity $\lambda_i$, whereas every $v\in V\backslash S$ has activity $\lambda$. Let $\mu_i$ be the Gibbs distribution on $G$ with parameters $\beta,\gamma, \lambdab_i$. Then, for every subgraph $F$ of $G$, it holds that 
\[\big|\Eb_{\sigma\sim \mu_2}[\oc_F(\sigma)]-\Eb_{\sigma\sim \mu_1}[\oc_F(\sigma)]\big|\leq 2\big(|a|+|b|+|c|\big)\big(|V(G)|^2+|E(G)|^2\big)\, \Big|\frac{\lambda_2}{\lambda_1}-1\Big|.\]
\end{lemma}

We now have all the ingredients to prove Lemma~\ref{lem:intermediateIsing}.
\begin{lemintermediateIsing}
\statelemintermediateIsing
\end{lemintermediateIsing}
\begin{proof}
Let $K=|a|+|b|+|c|$ and $q_\pl,q_\mi,\alpha$ be the constants from Lemma~\ref{lem:2spinconstruction}; recall that $\alpha\in (0,1)$ and $1>q_\pl>q_\mi>0$. Let $H$ be a cubic bipartite graph and $\hat\lambda>1$ be the inputs to $\MagIsing(\alpha)$, and let $\eta\in (0,1)$ be the desired relative error that we want to approximate $\Mc_{\alpha,\hat\lambda}(H)$.

By Lemma~\ref{lem:tcplus} and Theorem~\ref{thm:observablegadget}, there exist constants $\tilde R,\hat{R}, \hat{O},C>0$ with $\tilde R\neq 1$ and polynomial time algorithms, which on input rationals $r,r'\in(0,1/2)$, output in time  $poly(\bit(r),\bit(r'))$  pairs of field gadgets  $(\Tc_\pl,\Tc_\mi)$ and $(\Tc_1,\Tc_2)$ satisfying
\begin{equation}\label{eq:2wsfieldgadgets}
\begin{aligned}
R_{\pl}> R_{\mi}+Cr \mbox{ and } |R_{\pl} - \tilde R|,\ |R_{\mi} - \tilde R|\leq r,\\
|R_{1} - \hat{R}|,\, |R_{2} - \hat{R}|\leq r,\mbox{ but } |O_1-O_2|\geq \hat{O},
\end{aligned}
\end{equation}
where $R_\pl,R_\mi,R_1,R_2$ are the effective fields of $\Tc_\pl,\Tc_\mi,\Tc_1,\Tc_2$ and $O_1,O_2$ are the observable gaps of $\Tc_1,\Tc_2$, respectively.

Let $\epsilon=\tfrac{\eta}{|V(H)|^{8}}$ and $t=\big\lceil \big(\tfrac{1}{\epsilon}|V(H)|^2\log \hat\lambda\big)^6\big\rceil$. By Lemma~\ref{lem:slypotts}, there is an algorithm that outputs in time $poly(t,\tfrac{1}{\epsilon})$ an integer $n$ and a graph $G\in {\cal G}^{t,\epsilon}_{n,\Delta}$ (satisfying Items~\ref{it:Potts1} and~\ref{it:Potts2}).
%Moreover, the susceptibility gaps  $|S_{\Ec_1}|, |S_{\Ec_2}|$ are upper-bounded in absolute value by the constant $\Xi$, i.e., $|S_{\Ec_1}|, |S_{\Ec_2}|\leq \Xi$.
Let also $\Tc_\pl,\Tc_\mi$ be field gadgets  satisfying \eqref{eq:2wsfieldgadgets} for $r=\tfrac{|\tilde{R}-1|}{10}\epsilon^4$. Consider also the integers $\ell_\pl,\ell_\mi=\ell$, where $\ell$ is an integer specified according to whether $\hat\Lambda=\hat{\lambda}/\big(\tfrac{q_\pl \hat{R}+1-q_\pl}{q_\mi \hat{R}+1-q_\mi}\big)$ is bigger than 1. Suppose first that $\hat{\Lambda}\geq  1$.  Since $R_\pl>R_\mi$ and $q_\pl>q_\mi$, we have that $\tfrac{q_\pl R_{\pl}+1-q_\pl}{q_\mi R_{\pl}+1-q_\mi}>\tfrac{q_\pl R_{\mi}+1-q_\pl}{q_\mi R_{\mi}+1-q_\mi}$, and we pick   $\ell$ to be the smallest positive integer such that 
\begin{equation}\label{eq:lsmallest}
\big(\tfrac{q_\pl R+1-q_\pl}{q_\mi R+1-q_\mi}\big)\big(\tfrac{q_\pl R_{\pl}+1-q_\pl}{q_\mi R_{\pl}+1-q_\mi}\big)^{\ell}/\big(\tfrac{q_\pl R_{\mi}+1-q_\pl}{q_\mi R_{\mi}+1-q_\mi}\big)^{\ell}\geq \hat\lambda.
\end{equation}
If $\hat{\Lambda}< 1$, then we pick $\ell$ to be the small positive integer so that
\[\big(\tfrac{q_\pl R+1-q_\pl}{q_\mi R+1-q_\mi}\big)\big(\tfrac{q_\pl R_{\mi}+1-q_\pl}{q_\mi R_{\mi}+1-q_\mi}\big)^{\ell}/\big(\tfrac{q_\pl R_{\pl}+1-q_\pl}{q_\mi R_{\pl}+1-q_\mi}\big)^{\ell}\leq \hat\lambda,\] 
In either case, using the lower bound $R_{\pl}-R_{\mi}>Cr$ from \eqref{eq:2wsfieldgadgets}, we have that $\ell=O(\frac{1}{r}\log \hat\lambda)$  where the implicit constant depends only on $\beta,\gamma,\lambda,\Delta$. In particular, we have that $t\geq 5+\max\{\ell_\pl,\ell_\mi\}$. In the argument below, we assume w.l.o.g. that $\hat{\Lambda}\geq  1$; otherwise, just apply the same argument by swapping the roles of the gadgets $\Tc_\pl,\Tc_\mi$ in the construction below. 

For $i\in \{1,2\}$, consider now the graphs $\widehat{H}_i=H^{\ell_\pl,\ell_\mi}_{G,\Tc_\pl,\Tc_\mi,\Tc_i}$ and let $\mu_i=\mu_{\widehat{H}_i;\beta,\gamma,\lambda}$. For convenience, let also $F$ denote the graph $H^{\ell_\pl,\ell_\mi}_{G,\Tc_\pl,\Tc_\mi}$, and note that $F$ is a subgraph of both $\widehat{H}_1,\widehat{H}_2$. From Lemma~\ref{lem:2spinconstruction}, we have that for $i\in\{1,2\}$, for $\epsilon'=10 |V(H)| \epsilon$, it holds that
\begin{equation}\label{eq:4tgtg744}
\Oc_{\beta,\gamma,\lambda}(\widehat{H}_i)=\Ac_i|V(H)|+\Eb_{\sigma\sim \mu_i}\big[\oc_{F}(\sigma)\big]+(1\pm \epsilon') O_i \big[(q_\pl-q_\mi) \Mc_{\alpha,\hat\lambda_i}(H)+q_\mi|V(H)|\big],
\end{equation}
where  $\Ac_{i}= \Eb_{\sigma\sim \mu_{\Tc_i}}[\oc_{\Tc_i}(\sigma)\mid \sigma_{\rho_i}=0]$ and $\hat\lambda_i:=\big(\tfrac{q_\pl R_i+1-q_\pl}{q_\mi R_i+1-q_\mi}\big)\big(\tfrac{q_\pl R_{\pl}+1-q_\pl}{q_\mi R_{\pl}+1-q_\mi}\big)^{\ell_\pl}/\big(\tfrac{q_\pl R_{\mi}+1-q_\pl}{q_\mi R_{\mi}+1-q_\mi}\big)^{\ell_\mi}$.
From \eqref{eq:BguarP}, we have that $\hat\lambda_i=(1\pm \epsilon^3)\hat\lambda$, and therefore from Lemma~\ref{lem:2spinperturb}, we have that
\begin{gather*}
\big|\Eb_{\sigma\sim \mu_1}[\oc_{F}(\sigma)]-\Eb_{\sigma\sim \mu_2}[\oc_{F}(\sigma)]\big|\leq |E(H^\ell_{G,\Ec,\Pc})|\epsilon^3\leq \epsilon^2,\\
|\Mc_{\alpha,\hat\lambda_1}(H)-\Mc_{\alpha,\hat\lambda_2}(H)|\leq 2K\big(|V(H)|^2+|E(H)|^2\big)\epsilon^3\leq \epsilon^2.
\end{gather*}
We now invoke the oracle for $\Mc_{\beta,\gamma,\lambda}(\widehat{H}_i)$ to compute $\hat{\Mc}_i$ such that $\hat{\Mc}_i=(1\pm \epsilon^2)\Mc_{\beta,\gamma,\lambda}(\widehat{H}_i)$. By exploiting the tree structure of the field gadgets $\Tc_1,\Tc_2$ (cf. Definition~\ref{def:field}), and since they both have size $poly(\bit(r))$, we can compute the values $\Ac_1,\Ac_2$ exactly in time $poly(|V(H)|,\tfrac{1}{\eta})$ by fairly routine dynamic programming techniques.
Combining these with \eqref{eq:4tgtg744}, it follows that
\[\hat{M}=\frac{1}{q_\pl-q_\mi}\Big(\frac{(\hat{\Mc}_1-\hat{\Mc}_2)-|V(H)|(\Ac_1-\Ac_2)}{O_1-O_2}-q_\mi |E(H)|\Big)\]
is within a factor of $(1\pm \eta)$ of the susceptibility $\Mc_{\alpha,\hat\lambda}(H)$, as needed. This finishes the reduction and therefore the proof of Lemma~\ref{lem:intermediatePotts}. 
\end{proof}

\printbibliography

\appendix

\section{Proof of Theorem~\ref{thm:main2}}
\label{sec:proof-2spin}

\begin{proof}[Proof of Theorem~\ref{thm:main2}]
Let $\alpha\in (0,1)$ be as in Lemma~\ref{lem:intermediateIsing}. By \cite[Theorem 2]{CAI}, there is $\lambda_0>1$ such that for any $\lambda^*>\lambda_0$ the problem  $\IsingCubic(\alpha,\lambda^*)$, i.e., the problem of approximating the partition $Z_{G;\alpha,\alpha,\lambda^*}$ on cubic bipartite graphs $G$, is \#BIS-hard. From Lemma~\ref{lem:intermediateIsing}, we have that for $\beta>\beta_p(q,\Delta)$ it holds that
\[\MagIsing(\alpha)\leq_{\mathrm{AP}} \Obser(\beta,\gamma,\lambda,a,b,c).\]
so it suffices to show that $\IsingCubic(\alpha,\lambda^*)\leq_{\mathrm{AP}} \MagIsing(\alpha)$. For convenience, we will henceforth drop $\alpha$ from the subscript notation and write $Z_{G;\hat\lambda}, \Mc_{G;\alpha,\alpha,\hat\lambda}$ instead of $Z_{G;\alpha,\alpha,\hat\lambda}, \Mc_{G;\alpha,\alpha,\hat\lambda}$.

Let $G$ be an instance of $\IsingCubic(\alpha,\lambda^*)$ with $n$ vertices, and $\epsilon>0$ be the desired relative error for approximating $Z_{G;\lambda^*}$. Since $\frac{\partial \log Z_{G;\hat\lambda}}{\partial \hat\lambda}=\tfrac{1}{\hat\lambda}\Mc_{G;\hat\lambda}$, we have that
\begin{equation}\label{eq:interpolation2}
\log Z_{G;\lambda^*}=\int^{\lambda^*}_{1}\, \tfrac{1}{\hat\lambda}\Mc_{G;\hat\lambda} \mathrm{d}\hat\lambda.
\end{equation}
Let $T=\left\lceil \big(\tfrac{10\hat\lambda n}{\epsilon \alpha}\big)^4\right\rceil$ and for $i=0,1,\hdots, T$, consider the sequence of edge parameters $\hat\lambda_i=1+i\tfrac{\lambda^*-1}{T}$. The function $\log Z_{G;\hat\lambda}$ is convex with respect to $\hat\lambda$, since its second derivative is equal to the variance of the number of vertices with spin 1, and therefore the function $\tfrac{1}{\hat\lambda}\Mc_{G;\hat\lambda}$ is increasing. Since $G$ is cubic, we also have the crude bounds  $\tfrac{\alpha^{3}}{1+\alpha^{3}}n\leq \Mc_{G;\hat\lambda}\leq n$ that hold for all $\hat\lambda\geq 1$,\footnote{The upper bound is immediate, the lower bound can be proved by considering the spins of the neighborhood of a vertex, see for example \cite[Lemma 26]{SIDMA} for the case $\hat\lambda=1$ which is actually the worst-case for the argument therein.}, so approximating the integral in \eqref{eq:interpolation2} with a sum (analogously to the proof of Theorem~\ref{thm:main1}), we obtain that
\[ \log Z_{G;\lambda^*}=\big(1\pm\frac{\epsilon}{10}\big)\sum^T_{i=1}\tfrac{1}{\hat\lambda_i}\Mc_{G;\hat\lambda_i}.\]
Using the presumed oracle for $\MagIsing(\alpha)$, we compute $\hat{M}_i$ such that $\hat{M}_i=(1\pm \tfrac{\epsilon}{10})\Mc_{G;\hat\beta_i}$ for $i\in [T]$, and therefore the quantity $\hat Z=\exp\big(\sum_{i\in [T]} \tfrac{1}{\hat\lambda_i}\Mc_{G;\hat\lambda_i}\big)$
is a $(1\pm \epsilon)$-factor approximation to $Z_{G;\lambda^*}$. This completes the AP-reduction, and therefore the proof as well.
\end{proof}

\section{Proofs of Lemmas~\ref{lem:Pottsconstruction} and~\ref{lem:2spinconstruction}}\label{sec:deferred}
In this section, we give the deferred proofs of Lemmas~\ref{lem:Pottsconstruction} and~\ref{lem:2spinconstruction}. We start with Lemma~\ref{lem:Pottsconstruction}.
\begin{lemPottsconstruction}
\statelemPottsconstruction
\end{lemPottsconstruction}
\begin{proof}
For an edge $e=(u,v)\in E(H)$ and $i\in [\ell]$, let $\Pc^{i}_{e}$ be the $i$-th path connecting the gadgets $G_u,G_v$, respectively.   Denote by $\Ec_{e}$ the copy of $\Ec$ that connects $G_u,G_v$ and let $w_{u,e},w_{v,e}$ be roots of the gadgets $\Ec_{e}$.  Let also $W_{u,e},W_{u,v}$ be the vertices involved in the paths  $\Pc^{i}_{e}$ and the susceptibility gagdet $\Ec_{e}$ on each of $G_{u}, G_{v}$, respectively.

Note that the edge set of $H^{\ell}_{G,\Ec,\Pc}$ is the disjoint union of $E(H^{\ell}_{G, \Pc})$  and $E(\Ec_{e})$ for $e\in E(H)$, so we have 
\begin{equation}\label{eq:MHkGT}
\Sc_{q,\beta}(H^{\ell}_{G,\Ec,\Pc})=\Eb_{\sigma\sim \mu}[m_{H^\ell_{G,\Pc}}(\sigma)]+\sum_{e\in E(H)}\Eb_{\sigma\sim \mu}\Big[m_{\Ec_{e}}(\sigma)\Big].
\end{equation}
Fix arbitrary $e=(u,v)\in E(H)$ and for convenience let $\mu_{e}$ denote the Gibbs distribution on $\Ec_{e}$ with parameters $q,\beta$. Note that,  conditioned on the spins  of $w_{u,e},w_{v,e}$, the distribution $\mu$ factorizes; more precisely, for $\tau: V(\Ec_{e})\rightarrow [q]$ and $s_1,s_2\in [q]$, we have
\[\mu\Big(\sigma_{V(\Ec_{e})}=\tau\mid \sigma(w_{u,e})=s_1,\sigma(w_{v,e})=s_2\Big)=\mu_{e}(\tau\mid \sigma(w_{u,e})=s_1,\sigma(w_{v,e})=s_2).\] 
Using that the susceptibility gap of $\Ec_e$ is $S_\Ec$, and letting $\mathbf{1}_e$ to be the event that $\sigma(w_{u,e})=\sigma(w_{v,e})$  we therefore obtain that 
\begin{align*}
\Eb_{\sigma\sim \mu}\Big[m_{\Ec_{e}}(\sigma)\Big]&=\sum_{s_1,s_2\in [q]}\mu\big(\sigma(w_{u,e})=s_1,\sigma(w_{v,e})=s_2\big)\, \Eb_{\sigma\sim \mu_{e}}\big[m_{\Ec_{e}}(\sigma)\mid \sigma(w_{u,e})=s_1,\sigma(w_{v,e})=s_2\big]\\
&=S_\Ec\,\Eb_{\sigma\sim \mu}[\mathbf{1}_e]+\Ac_\Ec,
\end{align*}
where  the last equality follows by noting the symmetry among the colours and writing $\mu\big(\sigma(w_{u,e})\neq\sigma(w_{v,e})\big)=1-\Eb_{\sigma\sim \mu}[\mathbf{1}_e]$. Plugging into \eqref{eq:MHkGT}, we obtain
\begin{equation}\label{eq:4g4t4g56yttere5}
\Sc_{q,\beta}(H^{\ell}_{G,\Ec,\Pc})=\Ac_\Ec|E(H)|+\Eb_{\sigma\sim \mu}[m_{H^\ell_{G,\Pc}}(\sigma)]+S_\Ec \sum_{e\in E(H)}\Eb_{\sigma\sim \mu}[\mathbf{1}_e].
\end{equation}
For convenience, let $\overline{W}_e$ be the set $(W_u\cup W_v)\backslash \{w_{u,e},w_{v,e}\}$. We now expand $\Eb_{\sigma\sim \mu}[m(\sigma_{e})]$ according to the configuration  on the ports $\overline{W}_e$ and the phase assignment $Y$ on $V(H)$,   i.e., 
\begin{equation}\label{eq:4tg5tb5yhyhnh}
\Eb_{\sigma\sim \mu}[\mathbf{1}_e]=\sum_{\substack{Y: V(H)\rightarrow [q];\\\eta:\overline{W}_e\rightarrow [q]}}\mu\big(\sigma_{\overline{W}_e}=\eta, \widehat{\Yc}(\sigma)=Y\big)\, \Eb_{\sigma\sim \mu}[\mathbf{1}_e\mid \sigma_{\overline{W}_e}=\eta, \widehat{\Yc}(\sigma)=Y].
\end{equation}
Now, observe that $\overline{W}_e$ disconnects the vertices in $U_u\cup U_v$ from the rest of the graph $H^{\ell}_{G,\Ec,\Pc}$, so, conditioned on $\eta:\overline{W}_e\rightarrow [q]$, the configuration on $w_u,w_v$ is distributed according to $\mu_{G_u,G_v}$,  where $\mu_{G_u,G_v}$ is the Gibbs distribution on the subgraph of $H^{k}_{G,\Tc}$ induced by $U_u\cup U_v$. Therefore, for any $Y: V(H)\rightarrow [q]$ and $\eta:\overline{W}_e\rightarrow [q]$, we have that
\begin{equation}\label{eq:6N6J8R}
\Eb_{\sigma\sim \mu}[\mathbf{1}_e\mid \sigma_{\overline{W}_e}=\eta, \widehat{\Yc}(\sigma)=Y]=\Eb_{\sigma\sim \mu_{G_u,G_v}}[\mathbf{1}_e\mid \Yc(\sigma_{U_u})=Y(u), \Yc(\sigma_{U_v})=Y(v),\sigma_{\overline{W}_e}=\eta].
\end{equation}
Since $G_u,G_v$ are copies of $G$, and $G\in \Gc^{t,\epsilon}_{n,\Delta}$, by the product-distribution property in Item~\ref{it:Potts2}, we obtain that the expectation of $\mathbf{1}_e$ conditioned  on $\Yc(\sigma_{U_u})=Y(u), \Yc(\sigma_{U_v})=Y(v),\sigma_{\overline{W}_e}=\eta$ is approximately independent of $\eta$ and  within a factor of $(1\pm \epsilon)$ from the Gibbs distribution on $\Ec_e$ where $w_{u,e}$ gets the spin $Y_u$ with probability $p$ and any other spin with probability $\frac{1-p}{q-1}$, and similarly for $w_{v,e}$. To do the calculation of this expectation, let for convenience 
\[f_{ij}=\mu_{e}(\sigma(w_u)=i,\sigma(w_v)=j), \qquad f_{\equal}:=f_{11}, \qquad f_{\nequal}:=f_{12}\]
and note that by symmetry $f_{ii}=f_{\equal}$ for $i\in [q]$ and $f_{ij}=f_{\nequal}$ for distinct $i\neq j$. Using these, we can write $B_{\Ec}=\frac{ f_{\equal}}{ f_{\nequal}}$. The probability that $\mathbf{1}_e=1$ when $Y(u)=Y(v)=i$ is 
\begin{equation}\label{eq:Qequal}
%\frac{p^2f_{ii}+\tfrac{(1-p)^2}{(q-1)^2}\sum_{j\neq i}f_{jj}}{f_{ii}p^2+\tfrac{(1-p)^2}{(q-1)^2}\sum_{j\neq i}f_{jj}+2\tfrac{p(1-p)}{q-1}\sum_{j\neq i}f_{ij}+\sum_{j,j'\neq i}\tfrac{(1-p)^2}{(q-1)^2}f_{jj'}}=
\frac{R_0 f_{\equal}}{R_0 f_{\equal}+(1-R_0) f_\nequal}=A_0, \mbox{ where } R_0=p^2+\tfrac{(1-p)^2}{q-1}, \quad A_{0}=\tfrac{B_{\Ec}}{B_{\Ec}+\tfrac{1-R_0}{R_0}}.%\begin{array}{l} Q_\equal=p^2+\tfrac{(1-p)^2}{q-1}\\ Q_\nequal=\tfrac{2(q-2)p(1-p)}{q-1}+\tfrac{(q-2)^2(1-p)^2}{(q-1)^2}\end{array}
\end{equation} 
 In contrast,  when $Y(u)\neq Y(v)$, the probability that $\mathbf{1}_e=1$ is equal to 
\begin{equation}\label{eq:Qequalp}
\frac{R_1 f_{\equal}}{R_1 f_{\equal}+(1-R_1) f_\nequal}=A_1,  \mbox{ where } R_1=\tfrac{2p(1-p)}{q-1}+\tfrac{(q-2)(1-p)^2}{(q-1)^2},\quad A_{1}=\tfrac{B_{\Ec}}{B_{\Ec}+\tfrac{1-R_1}{R_1}}.%\begin{array}{l} Q_\equal'=\frac{2p(1-p)}{q-1}+\tfrac{(q-2)(1-p)^2}{q-1}\\ Q_\nequal'=p^2+2(q-2)\tfrac{p(1-p)}{q-1}+\frac{(q^2-3q+3)(1-p)^2}{(q-1)^2}\end{array}\] 
\end{equation}
It follows that
\[\Eb_{\sigma\sim \mu_{G_u,G_v}}[\mathbf{1}_e\mid \Yc(\sigma_{U_u})=Y(u), \Yc(\sigma_{U_v})=Y(v),\sigma_{\overline{W}_e}=\eta]=(1\pm \epsilon)A_{\mathbf{1}\{Y(u)\neq Y(v)\}}.\]
Combining this with \eqref{eq:6N6J8R}, and plugging back into \eqref{eq:4tg5tb5yhyhnh} we obtain that
\[\Eb_{\sigma\sim \mu}[\mathbf{1}_e]=(1\pm \epsilon)\sum_{Y: V(H)\rightarrow [q]}\mu\big(\widehat{\Yc}(\sigma)=Y\big) A_{\mathbf{1}\{Y(u)\neq Y(v)\}},\]
In turn, plugging this back into \eqref{eq:4tg5tb5yhyhnh} yields that
\begin{equation}\label{eq:4h}
\Sc_{q,\beta}(H^{\ell}_{G,\Ec,\Pc})=\Ac_\Ec|E(H)|+\Eb_{\sigma\sim \mu}[m_{H^\ell_{G,\Pc}}(\sigma)]+(1\pm \epsilon) S_\Ec\big[ (A_{0}-A_{1})\mbox{\textsc{Avg}}_\mu(H)+A_{\nequal}|E(H)|\big],
\end{equation}
where $\mbox{\textsc{Avg}}_\mu(H)=\sum_{Y: V(H)\rightarrow [q]}\mu\big(\widehat{\Yc}(\sigma)=Y\big)m_H(Y)$. We will show that for any phase vector $Y:V(H)\rightarrow [q]$, it holds that
\begin{equation}\label{eq:43534Pottsphases}
\mu\big(\widehat{\Yc}(\sigma)=Y\big)=(1\pm \tfrac{\epsilon'}{2})\mu_{H;q,\hat\beta}(Y).
\end{equation}
where $\hat\beta$ and $\epsilon'=10q|V(H)|$ are as in the statement of the lemma. From this, we conclude that $|\mbox{\textsc{Avg}}_\mu(H)-\Sc_{q,\hat\beta}(H)|\leq \frac{\epsilon'}{2}|V(H)|$, which completes the proof. 

It remains to prove \eqref{eq:43534Pottsphases}; this follows by adapting suitably the arguments in \cite{PottsBIS}. In particular, for an edge $e=(u,v)\in E(H)$ and $i\in [\ell]$, let $\Pc^{i}_{e}$ be the $i$-th path connecting the gadgets $G_u,G_v$, respectively, and $w^{i}_{u,e},w^{i}_{v,e}$ be the corresponding endpoints. Recall that $w_{u,e},w_{v,e}$ correspond instead to the endpoints of the gadget $\Ec$. Let $\hat{\mu}$ be the Gibbs distribution on the graph $H_G$, i.e., the graph where all the paths and susceptibility gadgets have been removed, so that $H_G$ consists of $V(H)$ disjoint copies of the gadget $G$. Then, for all $Y: V(H)\rightarrow [q]$, we have that  
\begin{equation}\label{eq:3fr54fff}
\begin{aligned}
\mu\big(\widehat{\Yc}(\sigma)&=Y\big)\propto \hat{\mu}\big(\widehat{\Yc}(\sigma)=Y\big)\times\\
&\sum_{\eta: W\rightarrow [q]} \hat{\mu}\big(\sigma_W=\eta\,\big|\, \widehat{\Yc}(\sigma)=Y)\prod_{e \in E(H)} (B_{\Ec})^{\mathbf{1}\{\sigma(w_{u,e})=\sigma(w_{v,e})\}} \prod_{i\in [\ell]} (B_{\Pc})^{\mathbf{1}\{\sigma(w^{i}_{u,e})=\sigma(w^{i}_{v,e})\}}.
\end{aligned}
\end{equation}
Since the gadgets $G_u\in G^{t,\epsilon}_{n,\Delta}$, we have from Items~\ref{it:Potts1} and~\ref{it:Potts2} that 
\[\hat{\mu}\big(\widehat{\Yc}(\sigma)=Y\big)=\big(\tfrac{1}{q}\pm \epsilon\big)^{|V_H|}\mbox{ and }\hat{\mu}\big(\sigma_W=\eta\, \big|\, \widehat{\Yc}(\sigma)=Y)=(1\pm \epsilon)^{|V(H)|}\prod_{v\in V(H)}Q^{Y(v)}_{W_v}(\eta_{W_v}).\]
It follows from \eqref{eq:3fr54fff} that for all $Y:V(H)\rightarrow[q]$
\begin{equation}\label{eq:3fr54fff2}
\begin{aligned}
\mu&\big(\widehat{\Yc}(\sigma)=Y\big)\propto\\
&(1\pm \tfrac{\epsilon'}{5})\sum_{\eta: W\rightarrow [q]} \prod_{v\in V(H)}Q^{Y(v)}_{W_v}(\eta_{W_v})\prod_{e \in E(H)} (B_{\Ec})^{\mathbf{1}\{\sigma(w_{u,e})=\sigma(w_{v,e})\}} \prod_{i\in [\ell]} (B_{\Pc})^{\mathbf{1}\{\sigma(w^{i}_{u,e})=\sigma(w^{i}_{v,e})\}}.
\end{aligned}
\end{equation}
For every edge $e=(u,v)\in V(H)$ such that $Y(u)=Y(v)$, for $i\in \ell$, the endpoints $w^i_{u,e}$ and $w^i_{v,e}$ of the path $\Pc^i_{e}$ for $i\in [\ell]$ have the same spin under the distributions $Q^{Y(u)}_{W_u}(\cdot)$ and $Q^{Y_u}_{W_u}(\cdot)$ with probability $R_0$ (as defined in \eqref{eq:Qequal}), and therefore they contribute a factor of $1+(B_{\Pc}-1)R_{\equal}$ to the sum; when $Y_u\neq Y_v$, the corresponding factor is given by $1+(B_{\Pc}-1)R_1$, with $R_1$ as defined in \eqref{eq:Qequalp}. Similarly, for $w_{u,e}$ and $w_{v,e}$ we obtain the expressions $1+(B_{\Ec}-1)R$ and $1+(B_{\Ec}-1)R'$, respectively. Therefore, for $\hat\beta:=\big(\frac{1+(B_{\Pc}-1)R_{\equal}}{1+(B_{\Pc}-1)R_{\nequal}}\big)^{\ell}\big(\tfrac{1+(B_{\Ec}-1)R_{\equal}}{1+(B_{\Ec}-1)R_{\nequal}}\big)$, it follows that
\[\mu\big(\widehat{\Yc}(\sigma)=Y\big)\propto(1\pm \tfrac{\epsilon'}{5})\hat\beta^{m_H(Y)} \mbox{ for all } Y: V(H)\rightarrow[q],\]
from which \eqref{eq:43534Pottsphases} follows. This finishes the proof of Lemma~\ref{lem:Pottsconstruction}.
\end{proof}
We next give the proof of Lemma~\ref{lem:2spinconstruction}, which builds upon analogous ideas.
\begin{lemspinconstruction}
\statelemspinconstruction
\end{lemspinconstruction}
\begin{proof}
By linearity of expectation, we have  
\begin{equation}\label{eq:bMbHkGT}
\Oc_{\beta,\gamma,\lambda}(H^{\ell_\pl,\ell_\mi}_{G,\Tc_\pl,\Tc_\mi,\Tc})=\Eb_{\sigma\sim \mu}\Big[\oc_{H^{\ell_\pl,\ell_\mi}_{G,\Tc_\pl,\Tc_\mi}}(\sigma)-a|\sigma_{W_{\Tc}}|\Big]+\sum_{v\in V(H)}\Eb_{\sigma\sim \mu}\big[\oc_{\Tc_{v}}(\sigma)\big],
\end{equation}
where we have subtracted the term $a|\sigma_{W_{\Tc}}|$ in the first expectation to avoid double-counting the contribution of the magnetization of the vertices in $W_\Tc$. 

Fix arbitrary $v\in V(H)$ and  let $\mu_{v}$ denote the Gibbs distribution on $\Tc_{v}$ with parameters $\beta,\gamma,\lambda$. Note that,  conditioned on the spin  of $w_{v}$, the distribution $\mu$ factorizes; more precisely, for $\tau: V(\Tc_{v})\rightarrow \two$ and $s\in \two$, we have $\mu\Big(\sigma_{V(\Tc_v)}=\tau\mid \sigma(w_{u})=s\Big)=\mu_{v}(\tau\mid \sigma(w_{v})=s)$. From this, and since the  observable gap of $\Tc_v$ is $M_\Tc$, and letting $\mathbf{1}_e$ to be the event that $\sigma(w_{u,e})=\sigma(w_{v,e})$  we therefore obtain that 
\begin{align*}
\Eb_{\sigma\sim \mu}\big[\oc_{\Tc_{v}}(\sigma)\big]&=\sum_{s\in \two}\mu\big(\sigma(w_{u})=s\big)\, \Eb_{\sigma\sim \mu_{v}}\big[\oc_{\Tc_{v}}(\sigma)\mid \sigma(w_{u})=s\big]=(O+a)\,\Eb_{\sigma\sim \mu}[\sigma(w_v)]+\Ac,
\end{align*}
where  the last equality follows by writing $\mu\big(\sigma(w_{u})=0\big)=1-\Eb_{\sigma\sim \mu}[\sigma(w_v)]$ and the definition of the observable gap (which accounts for the additive $a$ in the above formula). Plugging into \eqref{eq:bMbHkGT}, and by linearity of expectation, we obtain
\begin{equation}\label{eq:b4bg4t4g56yttere5}
\Oc_{\beta,\gamma,\lambda}(H^{\ell_\pl,\ell_\mi}_{G,\Tc_\pl,\Tc_\mi,\Tc})=\Ac|V(H)|+\Eb_{\sigma\sim \mu}\big[\oc_{H^{\ell_\pl,\ell_\mi}_{G,\Tc_\pl,\Tc_\mi}}(\sigma)\big]+O \sum_{v\in V(H)}\Eb_{\sigma\sim \mu}[\sigma(w_v)].
\end{equation}
For convenience, let $\overline{W}_v$ be the set $W_v\backslash \{w_{v}\}$. We expand $\Eb_{\sigma\sim \mu}[\sigma(w_v)]$ according to the configuration  on the ports $\overline{W}_v$ and the phase assignment $Y$ on $V(H)$,   i.e., 
\begin{equation}\label{eq:b4btg5tb5yhyhnh}
\Eb_{\sigma\sim \mu}[\sigma(w_v)]=\sum_{\substack{Y: V(H)\rightarrow \two;\\\eta:\overline{W}_v\rightarrow \two}}\mu\big(\sigma_{\overline{W}_v}=\eta, \widehat{\Yc}(\sigma)=Y\big)\, \Eb_{\sigma\sim \mu}[\sigma(w_v)\mid \sigma_{\overline{W}_v}=\eta, \widehat{\Yc}(\sigma)=Y].
\end{equation}
The set $\overline{W}_v$ disconnects the vertices in $U_u$ from the rest of the graph $H^{\ell_\pl,\ell_\mi}_{G,\Tc_\pl,\Tc_\mi,\Tc}$, and therefore conditioned on $\eta:\overline{W}_v\rightarrow [q]$, the configuration on $w_u$ is distributed according to $\mu_{G_u}$,  where $\mu_{G_u}$ is the Gibbs distribution on the phase gadget $G_u$. Therefore, for any $Y: V(H)\rightarrow \two$ and $\eta:\overline{W}_v\rightarrow \two$, we have that
\[\Eb_{\sigma\sim \mu}[\sigma(w_v)\mid \sigma_{\overline{W}_v}=\eta, \widehat{\Yc}(\sigma)=Y]=\Eb_{\sigma\sim \mu_{G_v}}[\sigma(w_v)\mid \sigma_{\overline{W}_v}=\eta, \Yc(\sigma)=Y(v)].\]
Since $G_v\in \Gc^{t,\epsilon}_{n,\Delta}$, by the product-distribution property in Item~\ref{it:2spin2}, we obtain that 
\[\Eb_{\sigma\sim \mu_{G_v}}[\sigma(w_v)\mid \sigma_{\overline{W}_v}=\eta, \Yc(\sigma)=Y(v)]=(1\pm \epsilon)q_{Y(v)},\]
where $q_{\plm}$ are as in Section~\ref{sec:phase2spingadgets}.
Plugging back into \eqref{eq:b4btg5tb5yhyhnh} we obtain that
\[\Eb_{\sigma\sim \mu}[\sigma(w_v)]=(1\pm \epsilon)\sum_{Y: V(H)\rightarrow \two}\mu\big(\widehat{\Yc}(\sigma)=Y\big) q_{Y(v)},\]
In turn, plugging this back into \eqref{eq:b4btg5tb5yhyhnh} yields that
\begin{equation}\label{eq:b5}
\Oc_{\beta,\gamma,\lambda}(H^{\ell_\pl,\ell_\mi}_{G,\Tc_\pl,\Tc_\mi,\Tc})=\Ac|V(H)|+\Eb_{\sigma\sim \mu}\big[\oc_{H^{\ell_\pl,\ell_\mi}_{G,\Tc_\pl,\Tc_\mi}}(\sigma)\big]+O \big[ (q_{\pl}-q_{\mi})\mbox{\textsc{Avg}}_\mu(H)+q_\mi |V(H)|\big].
\end{equation}
where $\mbox{\textsc{Avg}}_\mu(H)=\sum_{Y: V(H)\rightarrow \two}|Y|\mu\big(\widehat{\Yc}(\sigma)=Y\big)$, and $|Y|$ denotes the number of vertices $v\in V(H)$ with $Y(v)=\pl$. We will show that for any phase vector $Y:V(H)\rightarrow \two$, it holds that
\begin{equation}\label{eq:b4b35342spinphases}
\mu\big(\widehat{\Yc}(\sigma)=Y\big)=(1\pm \tfrac{\epsilon'}{2})\mu_{H;\alpha,\hat\lambda}(Y).
\end{equation}
where $\hat\lambda$ and $\epsilon'=10|V(H)|\epsilon$ are as in the statement of the lemma, and $\alpha\in (0,1)$ is a constant depending only on $\beta,\gamma,\lambda,\Delta$. From this, we conclude that $|\mbox{\textsc{Avg}}_\mu(H)-\Mc_{\alpha,\hat\lambda}(H)|\leq \frac{\epsilon'}{2}|V(H)|$, which completes the proof. 

To prove \eqref{eq:b4b35342spinphases}, we adapt arguments from \cite{CAI}. Let $W_H=\bigcup_{v\in V(H)} W_v$ denote the union of all ports over the phase gadgets $G$. Moreover, for a vertex $v\in  V(H)$ and $i\in [\ell_\pl]$, let $\hat\Tc^{i}_{\pl}$ be the $i$-th copy of the gadget $\Tc_\pl$ in $G_u$,  and $w^{i}_{v,\pl}$ be the corresponding root. Similarly, for $i\in [\ell_\mi]$, let $\Tc^{i}_{\mi}$ be the $i$-th copy of the field gadget $\Tc_\mi$ in $G_v$,  and $w^{i}_{v,\mi}$ be the corresponding root.   Recall also that $w_{v}$ is the root of the field gadget $\Tc_v$. For an edge $\{u,v\}\in E(H)$, denote also by $w^\pl_{u,e},w^\pl_{v,e}$ the endpoints of the edge connecting the $\pl$ sides of the gadgets $G_u,G_v$, and denote by $w^\mi_{u,e},w^\mi_{v,e}$ the corresponding endpoints for the $\mi$ side.

For a configuration $\eta:W_H\rightarrow\two$, $v\in V(H)$ and an edge $e=\{u,v\}\in E(H)$, it will be convenient to set
\begin{equation*}
\begin{aligned}
\Pi_v(\eta)&:= R_\Tc^{\mathbf{1}\{\eta(w_{u})=1\}}\mbox{$\prod_{i\in [\ell_\pl]}$\,}(R_{\Tc_\pl})^{\mathbf{1}\{\eta(w_{u,\pl}^i)=1\}} \mbox{$\prod_{i\in [\ell_\mi]}$\,}(R_{\Tc_\mi})^{\mathbf{1}\{\eta(w_{v,\mi}^{i})=1\}},\\
\Pi_{e}(\eta)&=\mbox{$\prod_{i\in {\pl,\mi}}$\,}\beta^{\mathbf{1}\{\sigma(w^i_{u,e})=\sigma(w^i_{v,e})=0\}}\,\gamma^{\mathbf{1}\{\sigma(w^i_{u,e})=\sigma(w^i_{v,e})=1\}}.
\end{aligned}
\end{equation*}
Let $\hat{\mu}$ be the Gibbs distribution with parameters $\beta,\gamma,\lambda$ on the graph $H_G$, i.e., the graph consisting only of the phase gadgets $G$, where all the field gadgets have been removed. Then, for all $Y: V(H)\rightarrow \two$, we have that  
\begin{equation}\label{eq:aab3bfr54fff}
\begin{aligned}
\mu\big(\widehat{\Yc}(\sigma)&=Y\big)\propto \hat{\mu}\big(\widehat{\Yc}(\sigma)=Y\big)\sum_{\eta: W_H\rightarrow \two} \hat{\mu}\big(\sigma_W=\eta\,\big|\, \widehat{\Yc}(\sigma)=Y)\prod_{v\in V(H)}\Pi_v(\eta)\prod_{e\in E(H)}\Pi_e(\eta).
\end{aligned}
\end{equation}
Since the gadgets $G_v\in G^{t,\epsilon}_{n,\Delta}$ for all $v\in V(H)$, we have from Items~\ref{it:2spin1} and~\ref{it:2spin2} that 
\[\hat{\mu}\big(\widehat{\Yc}(\sigma)=Y\big)=\big(\tfrac{1}{2}\pm \epsilon\big)^{|V_H|}\ \mbox{ and }\ \hat{\mu}\big(\sigma_W=\eta\, \big|\, \widehat{\Yc}(\sigma)=Y)=(1\pm \epsilon)^{|V(H)|}\prod_{v\in V(H)}Q^{Y(v)}_{W_v}(\eta_{W_v}).\]
Therefore we can rewrite \eqref{eq:aab3bfr54fff} as 
\begin{equation}\label{eq:bbbger3fr54fff2}
\begin{aligned}
\mu\big(\widehat{\Yc}(\sigma)&=Y\big)\propto \big(1\pm \tfrac{\epsilon'}{5}\big)\sum_{\eta: W_H\rightarrow \two} \prod_{v\in V(H)}Q^{Y(v)}_{W_v}(\eta_{W_v})\prod_{v\in V(H)}\Pi_v(\eta)\prod_{e\in E(H)}\Pi_e(\eta).
\end{aligned}
\end{equation}
For every vertex $v\in V(H)$, the multiplicative contribution to the sum in the r.h.s. of \eqref{eq:bbbger3fr54fff2} coming from $v$  is 
\begin{align*}
\big(q_\pl R_{\pl}+1-q_\pl\big)^{\ell_\pl}\big(q_\mi R_{\mi}+1-q_\mi\big)^{\ell_\mi}\big(q_\pl R+1-q_\pl\big) &\quad \mbox{ if $Y(v)=\pl$},\\
\big(q_\mi R_{\pl}+1-q_\mi\big)^{\ell_\pl}\big(q_\pl R_{\mi}+1-q_\pl\big)^{\ell_\mi}\big(q_\mi R+1-q_\mi\big)&\quad  \mbox{ if $Y(v)=\mi$.}
\end{align*}
To calculate the multiplicative contribution for an edge $e=\{u,v\}\in E(H)$, it will be convenient to define the matrix $\Mb=\{M_{ij}\}_{i,j\in \plm}$ from $\Mb=\big[\begin{smallmatrix} 1-q_\mi& q_\mi\\1-q_\pl&q_\pl\end{smallmatrix}\big]\big[\begin{smallmatrix} \beta&1\\1&\gamma\end{smallmatrix}\big]\big[\begin{smallmatrix} 1-q_\mi& 1-q_\pl\\q_\mi&q_\pl\end{smallmatrix}\big]$. Then, the multiplicative contribution to the sum in the r.h.s. of \eqref{eq:bbbger3fr54fff2} coming from an edge $e=\{u,v\}\in E(H)$ is  
\[M_{\pl\pl}M_{\mi\mi} \quad \mbox{ if $Y(u)=Y(v)$},\qquad  M_{\pl\mi}M_{\mi\pl}\quad \mbox{ if $Y(u)\neq Y(v)$.}\]
Using that $\beta\gamma<1$ and that $q_\pl\neq q_\mi$, it follows that $\alpha:=\tfrac{M_{\pl\pl}M_{\mi\mi}}{M_{\pl\mi}M_{\mi\pl}}\in (0,1)$.  Therefore, for $\hat\lambda:=\big(\tfrac{q_\pl R+1-q_\pl}{q_\mi R+1-q_\mi}\big)\big(\tfrac{q_\pl R_{\pl}+1-q_\pl}{q_\mi R_{\pl}+1-q_\mi}\big)^{\ell_\pl}/\big(\tfrac{q_\pl R_{\mi}+1-q_\pl}{q_\mi R_{\mi}+1-q_\mi}\big)^{\ell_\mi}$ as in the statement of the lemma, it follows that
\[\mu\big(\widehat{\Yc}(\sigma)=Y\big)\propto(1\pm \tfrac{\epsilon'}{5})\alpha^{m_H(Y)}\hat\lambda^{|Y|} \mbox{ for all } Y: V(H)\rightarrow\two,\]
which yields \eqref{eq:b4b35342spinphases}, and therefore completes the proof of Lemma~\ref{lem:2spinconstruction}.
\end{proof}

\section{Constructing the field  and edge-interaction gadgets}

\subsection{Recursive Constructions}

\subsubsection{Field Gadgets with observable gaps for 2-spin systems}
The following lemma gives the effective field and observable gap of a field gadget built out of smaller field gadgets. An analogous lemma was proved in \cite{averages} in the special case of magnetization, but the proof therein crucially uses the fact that the magnetization can be viewed as a partial derivative of the partition function. Here, we show a more general argument that does not require this partial-derivative view (which is not true for general vertex-edge observables).
\begin{lemma}\label{lem:2spinrecursion}
Let $(\beta,\gamma)$ be antiferromagnetic and $\lambda>0$, and let $(a,b,c)$ be a vertex-edge observable. 

Suppose we have field gadgets $\Tc_1,\dots,\Tc_k$ with roots $\rho_1,\hdots,\rho_k$, effective fields $R_1,\dots,R_k$ and
observable gaps $O_1,\dots,O_k$, respectively. Let $\Tc$ be the field gadget with root $\rho$, constructed by identifying the 
roots of $\Tc_1,\dots,\Tc_k$ into a single vertex $u$ and adding the edge $\{\rho,u\}$. Let $R$ and $O$ be the
effective field and observable gap for $\Tc$. Then
\begin{equation*}
R = \frac{1+\gamma\lambda\prod_{i\in [k]} R_i}{\beta+\lambda \prod_{i\in [k]}R_i}, \qquad O  = \theta(R) - \omega(R)  \sum_{i\in[k]} O_i,
\end{equation*}
where $\omega(R) := \frac{1+\beta \gamma-\beta R-\gamma/R}{1-\beta \gamma}$ and  $\theta(R)=-a\cdot \omega (R)- b \frac{\beta(R-\gamma)}{1-\beta \gamma}+c\frac{\gamma(1/R-\beta)}{1-\beta\gamma}$. Moreover, it holds that $R\in (\gamma, 1/\beta)$ and $0<\omega(R)<1$.
\end{lemma}
\begin{proof}[Proof of Lemma~\ref{lem:2spinrecursion}]
The equality for $R$ has been shown in \cite{averages}, though it will be helpful to review the derivation.  Let $Z_1,Z_0$ be the weight of all configurations on $\Tc$  where $\rho$ has spins $1$ and 0 respectively, and define analogously $Z_{1,i},Z_{0,i}$ for $i=1,\hdots,k$. Then, we have
\begin{align}\label{eq:Z0Z1}
Z_1=\lambda \bigg(\prod_{i\in [k]} Z_{0,i}+\gamma\lambda\prod_{i\in [k]} \frac{Z_{i,1}}{\lambda}\bigg), \quad
Z_0=\beta\prod_{i\in [k]} Z_{0,i}+\lambda \prod_{i\in [k]} \frac{Z_{i,1}}{\lambda}
\end{align}
We can therefore express the effective field of $\Tc$ as 
\begin{equation}\label{eq:R1R1RR}
R=\frac{1}{\lambda}\frac{Z_1}{Z_0}=\frac{1+\gamma\lambda\prod^{k}_{i=1} R_i}{\beta+\lambda \prod_{i\in [k]} R_i}, \mbox{ and hence } \prod_{i\in [k]}R_i=\frac{1-\beta R}{\lambda(R-\gamma)}.
\end{equation}
To find the observable gap, it will be convenient to  write $\mu$ instead of $\mu_{\Tc;\beta,\gamma,\lambda}$ and  $\mu_i$ instead of $\mu_{\Tc_i;\beta,\gamma,\lambda}$. For $s,t\in\{0,1\}$, let $p_{st}=\mu(\sigma(u)=t\mid \sigma(\rho)=s)$. Using \eqref{eq:Z0Z1} and \eqref{eq:R1R1RR}, we have the more precise expressions
\begin{equation}\label{eq:p0001}
\begin{aligned}
p_{00}&=\frac{\beta\prod^{k}_{i=1} Z_{0,i}}{Z_0}=\frac{\beta}{\beta+\lambda \prod^k_{i=1}R_i}=\frac{\beta(R-\gamma)}{1-\beta \gamma},&&  p_{01}=1-p_{00}=\frac{1-\beta R}{1-\beta \gamma}, \\
p_{10}&=\frac{\lambda\prod^{k}_{i=1} Z_{0,i}}{Z_1}=\frac{1}{1+\gamma \lambda \prod^{k}_{i=1} R_i}=\frac{1-\gamma/R}{1-\beta\gamma},&& p_{11}=1-p_{10}=\frac{\gamma(1/R-\beta)}{1-\beta\gamma}.
\end{aligned}
\end{equation}
Observe now that, conditioned on $\sigma(u)$, the distribution $\mu$ factorizes, and  the configurations on $\Tc_1,\hdots,\Tc_k$ are independent. Therefore
\begin{align*}
\Eb_{\sigma\sim\mu}&[\oc_{\Tc}(\sigma)\mid \sigma_\rho=0]=\sum_{t=0,1}p_{0t}\Eb_{\sigma\sim\mu}[\oc_{\Tc}(\sigma)\mid \sigma_u=t,\sigma_\rho=0]\\
&=p_{00}\Big(b+\sum_{i\in [k]}\Eb_{\sigma\sim\mu_{i}}[\oc_{\Tc_i}(\sigma)\mid \sigma_{\rho_i}=0]\Big) +p_{01}\Big(-a(k-1)+\sum_{i\in [k]}\Eb_{\sigma\sim\mu_{i}}\big[\oc_{\Tc_i}(\sigma)\mid \sigma_{\rho_i}=1\big]\Big),
\end{align*}
and similarly
\begin{align*}
\Eb_{\sigma\sim\mu}&[\oc_\Tc(\sigma)\mid \sigma_\rho=1]=\sum_{t=0,1}p_{1t}\Eb_{\sigma\sim\mu}[\oc_{\Tc}(\sigma)\mid \sigma_u=t,\sigma_\rho=1]\\
&=a+p_{10}\sum_{i\in [k]}\Eb_{\sigma\sim\mu_{i}}[\oc_{\Tc_i}(\sigma)\mid \sigma_{\rho_i}=0] +p_{11}\Big(c-a(k-1)+\sum_{i\in [k]}\Eb_{\sigma\sim\mu_{i}}\big[\oc_{\Tc_i}(\sigma)\mid \sigma_{\rho_i}=1\big]\Big).
\end{align*}
Note that $p_{11}-p_{01}=p_{00}-p_{10}=-\omega(R)$ and $-a\cdot \omega(R)+p_{11}c-p_{00}b=\theta(R)$, cf. \eqref{eq:p0001}. Since 
$ O=\Eb_{\sigma\sim\mu}[\oc_\Tc(\sigma)\mid \sigma_\rho=1]-a-\Eb_{\sigma\sim\mu}  [\oc_\Tc(\sigma)\mid \sigma_\rho=0]$, by subtracting the equations above we obtain the expression for $O$. 
\end{proof}

\subsubsection{Edge-interaction gadgets with susceptibility gaps for the Potts model}\label{sec:susc1Potts}
The following lemma is the analogue of Lemma~\ref{lem:2spinrecursion} in the case of the Potts model.
\begin{lemma}\label{lem:Pottsrecursion}
Let $q\geq 3$ and $\beta>1$. Suppose we have edge-interaction gadgets $\Ec_1,\dots,\Ec_k$ with ports $(\rho_1,\rho_1'),\hdots,(\rho_k,\rho_k')$, effective edge activities $B_1,\dots,B_k$ and
susceptibility gaps $S_1,\dots,S_k$, respectively. Let $\Ec$ be the edge-interaction gadget with ports $(\rho,\rho')$, constructed by identifying the 
ports $\rho_1,\hdots,\rho_k$ into a single vertex $u$, the ports $\rho_1',\hdots,\rho_k'$ into a single vertex $v$, and adding the edges $\{\rho,u\}$ and $\{v,\rho'\}$. Let $B$ and $S$ be the
effective edge activity and susceptibility gap for $\Ec$, respectively. Then
\begin{equation*}
B = \frac{1+\hat\gamma\hat\lambda\prod_{i\in [k]} B_i}{\hat\beta+\hat\lambda \prod_{i\in [k]}B_i}, \qquad S  = \theta(B) - \omega(B)  \sum_{i\in[k]} O_i,
\end{equation*}
where $\hat\beta:=1+\tfrac{(\beta-1)^2}{(q-1)(2\beta+q-2)}$, $\hat\gamma:=1+\tfrac{(\beta-1)^2}{2\beta+q-2}$, $\hat\lambda:=\tfrac{1-\hat\beta}{1-\hat\gamma}=\tfrac{1}{q-1}$, and $\omega(B) := \tfrac{1+\hat\beta \hat\gamma-\hat\beta B-\hat\gamma/B}{1-\hat\beta \hat\gamma}$, $\theta(B):=\tfrac{2\beta ( B-1) (B + q-1)}{B (\beta-1) (\beta + q-1)}$. Moreover, it holds that $B\in (1, \gamma)$ and $-1<\omega(B)<0$.
\end{lemma}
\begin{proof}
To simplify notation, we will write $\mu$ instead of $\mu_{\Ec;q,\beta}$, and similarly for $\Ec_1,\hdots,\Ec_k$. Let $Z_0$ be the weight of all configurations on $\Ec$ where $\rho,\rho'$ have the  spin 1, and $Z_1$ be the weight of all configurations on $\Ec$  where $\rho,\rho'$ have spins 1 and 2, respectively. Define analogously $Z_{0,i},Z_{1,i}$ for the edge-interaction gadget $\Ec_{i}$, $i\in [k]$. Then, we have the equalities
\begin{equation}\label{eq:Zin5t4Zout}
\begin{aligned}
Z_0 &=\big(\beta^2+q-1\big)\mbox{$\prod_{i\in [k]}$\,} Z_{0,i}+\big(2(q-1)\beta+(q-1)(q-2)\big)\mbox{$\prod_{i\in [k]}$\,}Z_{1,i},\\ 
Z_1 &=\big(2\beta+q-2\big)\mbox{$\prod_{i\in [k]}$\,} Z_{0,i}+\big(\beta^2+2(q-2)\beta+(q^2-3q+3)\big)\mbox{$\prod_{i\in [k]}$\,}Z_{1,i},\\
\end{aligned}
\end{equation}
from which the expression for $B=\frac{Z_0}{Z_1}=\tfrac{1+\hat\gamma\hat\lambda\prod_{i\in [k]} B_i}{\hat\beta+\hat\lambda \prod_{i\in [k]}B_i}$ given in the statement follows, using that $B_{i}=\frac{Z_{0,i}}{Z_{1,i}}$ for $i\in [k]$.

Note that $S=\beta\tfrac{\partial\log B}{\partial \beta}$ and  $S_i=\beta\tfrac{\partial\log B}{\partial \beta}=\tfrac{\beta}{B_i}\tfrac{\partial B_i}{\partial \beta}$. Using the expression for $B$, it follows that
\begin{equation}\label{eq:tggggg5g6g6sfww}
S/\beta=\frac{\hat\gamma\hat\lambda\prod_{i\in [k]} B_i\big(\sum_{i\in [k]}S_i/\beta\big)}{1+\hat\gamma\hat\lambda\prod_{i\in [k]} B_i}+\frac{\frac{\partial \hat\gamma}{\partial\beta}\hat\lambda\prod_{i\in [k]} B_i}{1+\hat\gamma\hat\lambda\prod_{i\in [k]} B_i}-\frac{\frac{\partial \hat\beta}{\partial\beta}+\hat\lambda\prod_{i\in [k]} B_i\big(\sum_{i\in [k]}S_i/\beta\big)}{\hat\beta+\hat\lambda \prod_{i\in [k]}B_i}
\end{equation}
From $B=\tfrac{1+\hat\gamma\hat\lambda\prod_{i\in [k]} B_i}{\hat\beta+\hat\lambda \prod_{i\in [k]}B_i}$ and  $\hat\lambda=\tfrac{1-\hat\beta}{1-\hat\gamma}$, we have that $\prod_{i\in [k]} B_i=\frac{\hat\beta B-1}{(\hat\gamma-B) \hat\lambda}=\frac{(B \hat\beta-1) (\hat\gamma-1)}{( \hat\beta-1) (B - \hat\gamma)}$. We also have that $\frac{\partial \hat\beta}{\partial\beta}=\hat\lambda\frac{\partial \hat\gamma}{\partial\beta}$ and therefore we can rewrite \eqref{eq:tggggg5g6g6sfww} as
\[S/\beta=-\omega(B)\Big(\sum_{i\in [k]}S_i/\beta\Big)+\theta(B)/\beta,\]
which can be verifed by just plugging everything in. This yields the expression for $S$, thus finishing the proof of the lemma.
\end{proof}

We are now in position to give the proof of the first of the interaction-gadget lemmas for the Potts model.
\begin{lemintergadget}
\statelemintergadget
\end{lemintergadget}
\begin{proof}
For an integer $\ell\geq 0$, let $B_\ell$ be the edge-interaction of the path with $2\ell+1$ edges. We have $B_0=\beta>1$ and, by Lemma~\ref{lem:Pottsrecursion} (applied with $k=1$), we have that $B_{\ell}=\frac{1+\hat\gamma\hat\lambda B_{\ell-1}}{\hat\beta+\hat\lambda B_{\ell-1}}$ where $\hat\beta,\hat\gamma,\hat\lambda$ are as in Lemma~\ref{lem:Pottsrecursion}. Using that $\hat\lambda=\tfrac{1-\hat\beta}{1-\gamma}$, it follows that 
\[B_\ell-1=\frac{(\hat\beta-1)(\hat\gamma -1)}{\hat\beta\hat\gamma-1}(B_{\ell-1}-1).\]
Since $\hat\beta,\hat\gamma>1$, we have that $\kappa:=\tfrac{(\hat\beta-1)(\hat\gamma -1)}{\hat\beta\hat\gamma-1}$ is a constant satisfying $0<\kappa<1$, and therefore for $\ell=\lceil \frac{\log r}{\log \kappa}\rceil$, we have $0<B_{\ell}-1<r$, as wanted.
\end{proof}
\subsection{Gadgets with dense fields and edge activities around a fixpoint}\label{sec:densefixpoint}

The first step in the construction is to create field and edge-interaction gadgets whose  effective fields and edge activities, respectively, are sufficiently dense in a small interval. The following quantities $x^*=x^*$ and $\omega^*=\omega^*_{\twospin}$ will be relevant for our arguments.
\begin{equation}\label{ode2spin}
x^* = \frac{1+\gamma\lambda (x^*)^2}{\beta+\lambda (x^*) ^2},\quad \omega^*_\twospin= \omega(x^*) = \frac{1+\beta \gamma-\beta x^*-\gamma/x^*}{1-\beta \gamma},
\end{equation}
and note  that $\omega^*_\twospin\in(0,1)$ from Lemma~\ref{lem:2spinrecursion}. The following lemma is taken from \cite{averages} (which in turn is based on techniques from~\cite{complex}).
\begin{lemma}[\mbox{\cite[Lemma 18]{averages}}]\label{ledtwospin}
For any $\delta>0$ and any sufficiently small $\tau_1>0$, there is constant $\tau\in (0,\tau_1)$ and a set of field gadgets $\mathcal{L}_\twospin=\{{\cal T}_1,\dots,{\cal T}_k\}$ with effective fields in the interval $I''_{\twospin}:=[x^* - \tau,x^*+\tau]$ and which is such that,  for any
$x\in [x^* - \tau,x^*+\tau]$, there exists $\Tc\in \mathcal{L}_\twospin$ whose effective field is in the
interval $[x-\tau\delta,x+\tau\delta]$.
\end{lemma}

We have the following analogue for the Potts model, using exactly the same argument as in \cite{averages}, and utilising that 1 is a fixpoint of the recursions in Lemma~\ref{lem:Pottsrecursion}. 
\begin{lemma}\label{ledPotts}
For any $\delta>0$ and any sufficiently small constant $\tau_1>0$, there is constant $\tau\in (0,\tau_1)$ and a set of edge-interaction gadgets $\mathcal{L}_{\Potts}=\{\Ec_1,\dots,\Ec_k\}$ with edge interactions in the interval $I''_{\Potts}:=(1,1+\tau)$ and which is such that, for any
$x\in (1,1+\tau)$, there exists $\Ec\in \mathcal{L}_{\Potts}$ whose edge interaction is in the
interval $[x-\tau\delta,x+\tau\delta]$.
\end{lemma}
Later, we will also need the following quantity $\omega^*_{\Potts}$ given by
\[\omega^*_\Potts= \omega(1) = \frac{1+\hat\beta \hat\gamma-\hat\beta-\hat\gamma}{1-\hat\beta \hat\gamma}=-\frac{(\hat\beta-1)(\hat\gamma-1)}{\hat\beta\hat\gamma-1}\]
and note  that $|\omega^*_\Potts|\in(0,1)$ from Lemma~\ref{lem:Pottsrecursion}.

\subsection{Obtaining non-zero susceptibility and observable gaps}\label{sec:lastproofs}
The second step is to use in an iterative way the field and edge-interaction gadgets in the sets ${\cal L}_\twospin,{\cal L}_\Potts$ constructed in Section~\ref{sec:densefixpoint}. We start by  giving the details for 2-spin models, the details for the Potts model are analogous and we sketch the main differences in the end of this section. Henceforth, for $s\in\{\twospin,\Potts\}$, we will use $[L_s]$ as an index set for the gadgets in $\{{\cal L}_s\}$, e.g., for $i\in [L_{\twospin}]$, $\Tc_i$ will denote the $i$-th field gadget in ${\cal L}_\twospin$.

We will start with some (arbitrary) field gadget $\Tc_0$. Suppose that at some stage we have constructed a field gadget $\Tc$ with effective field $R$ and observable gap $O$. To construct a new field gadget, we merge $\Tc$ with one field gadget from ${\cal L}$ using the operation from Lemma~\ref{lem:2spinrecursion}. We are going to analyze what pairs of
(effective field, observable gap) are achievable by field gadgets constructed using this procedure. For $i\in [L_\twospin]$, let $R_i$ and $O_i$ be the effective
field and observable gap for the $i$-th field gadget in our collection ${\cal L}_{\twospin}$. By Lemma~\ref{lem:2spinrecursion}, merging the field gadget $\Tc$ with
$\Tc_i$ yields a rooted field gadget with effective field $\phi_{i}(R)$ and observable gap $\psi_{i}(R,O)$ where the pair of maps $(\phi_i,\psi_i)$ are given by
\begin{equation}\label{eq:1qazw2spin}
\phi_{i}(R) = \tfrac{1+\gamma \lambda R R_i}{\beta +\lambda R R_i}, \qquad
\psi_{i}(R,O)= \theta\big(\phi_{i}(R)\big) - \omega \big(\phi_{i}(R)\big)( O + O_i),
\end{equation}
where $\omega(\cdot),\theta(\cdot)$ are the functions specified in Lemma~\ref{lem:2spinrecursion}. 
For a small constant $\tau>0$ to be specified later, let
\begin{equation}\label{idef2twospinPotts}
I' = \Big[x^* - \tau \tfrac{2|\omega^*|}{1-|\omega^*|}, x^* +  \tau \tfrac{2|\omega^*|}{1-|\omega^*|}\Big].
\end{equation}
The choice of the interval $I'$ is such that the maps $\phi_i$ are uniformly contracting and map the interval $I'$ to itself. Namely, we will use the following lemma from \cite{averages}.
\begin{lemma}[\mbox{\cite[Lemma 19]{averages}}]\label{contr}
There exist $0<C_{\mathrm{min}} < C_{\mathrm{max}}<1$ and $\tau_0>0$ such that for all $\tau\in(0,\tau_0)$
and all $i\in [L_{\twospin}]$ it holds that
\[|\phi'_{i}(I)|,|\omega(I')|\subseteq[C_{\mathrm{min}},  C_{\mathrm{max}}], \quad \phi_{i}(I')\subseteq I'.\]
\end{lemma}
The iterative construction actually takes place in the following smaller sub-interval of $I'$:
\begin{equation}\label{idef}
I= \Big[x^* - \tau\tfrac{|\omega^*|}{2}, x^* + \tau \tfrac{|\omega^*|}{2}\Big] \subseteq I'.
\end{equation}
The choice of $I$ is to ensure the following ``well-covered'' property for the corresponding maps $\phi_i$, which is ultimately due to the  density of the family ${\cal L}_\twospin$ from Lemma~\ref{ledtwospin}.
\begin{lemma}[\mbox{\cite[Lemma 20]{averages}}]\label{lctg4}
Suppose $\delta<|\omega^*|/100$. For every two points $x_1,x_2\in I$ such that $|x_1-x_2|\leq |\omega^*|\tau \delta/2$
there exists $i$ such that $x_1,x_2\in\phi_{i}(I)$.
\end{lemma}

Lemmas~\ref{contr} and~\ref{lctg4} ensure that, for any $x\in I$, a sequence of gadgets can be constructed whose effective field converges (quickly) to $x$. Namely, consider the following process \verb!Build-gadget!$(x,t)$ where $x\in I$ and $t\geq 0$ is an integer. If $t=0$
we return the degenerate gadget. If $t\geq 1$ then we let $\phi_{i}$ be any map such that
$x\in\phi_{i,s}(I)_s$. Let $y = \phi_{i}^{-1}(x)$ and  $\Tc'=$\verb!Build-gadget!$(y,t-1)$. Return the gadget obtained by merging $\Tc'$ and $\Tc_i$ using the operation of Lemma~\ref{lem:2spinrecursion}. 

The point behind the process \verb!Build-gadget!$(x,t)$ is that it yields, for arbitrary $x\in I$, a field gadget whose effective field is arbitrary close to $x$, with error that decays exponentially fast with $t$ (using the contraction properties of the $\phi_i$'s). 
\begin{lemma}\label{yyyqw3r4e}
There exist $C,C'>0$ such that for any $x\in I$ and any integer $t\geq 0$
the effective field $R$ of the field gadget returned by \verb!Build-gadget!$(x,t)$
(for any choice of the $\phi_i$'s) satisfies  $|R-x|\leq C C_{\mathrm{max}}^t$.
\end{lemma}

For the field gadgets constructed through Lemma~\ref{yyyqw3r4e}, the effective field will always be the interval $I'$, from Lemma~\ref{contr}. It is not hard to see that the observable gaps of the field gadgets constructed using our process also stay restricted to an interval $J$. Let
\begin{equation}\label{eq:obsbounds}
\hat{T}:= \max_{R\in I'} |\theta(R)|, \quad T:= \frac{\hat{T}+\max|O_i|}{1-C_{\mathrm{max}}}, \mbox{ and $J$ be the interval $[-T,T]$.}
\end{equation}
\begin{lemma}\label{lem:obsbounds}
Suppose $R\in I'$ and $O\in J$ then $\psi_i(R,O)\in J$.
\end{lemma}
\begin{proof}
We have $\psi_i(R,O)  = \theta(R) - \omega(R)( O_i +O)$.  From~\eqref{eq:obsbounds}, we have that $|\psi_i(R,O)|\leq \hat{T} + C_{\mathrm{max}} ( T + \max |O_i|) \leq T$.
\end{proof}

For any $x\in I$ we are going to construct a sequence of families of pairs of maps
${\cal F}_{x,0},{\cal F}_{x,1},\dots$ as follows. In each pair, the first map is from $\mathbb{R}$ to $\mathbb{R}$,
and the second map is from $\mathbb{R}^2$ to $\mathbb{R}$ (similarly to \eqref{eq:1qazw2spin}). Let ${\cal F}_{x,0}$ contain the pair of
maps $(x \mapsto x,(x,y)\mapsto y)$. To construct ${\cal F}_{x,t+1}$ we take every $\phi_i$
such that that $x\in \phi_i(I)$ and every $(f,g)\in {\cal F}_{\phi_i^{-1}(x),t}$ and place into ${\cal F}_{x,t+1}$
the map  
\begin{equation}\label{eq:remark}
(R,O)\mapsto \left( \phi_i(f(R)), \psi_i(f(R),g(R,O)) \right).
\end{equation}
There are two possible ``cases'' for the families ${\cal F}_{x,t}$. In the first case
the gadgets we need are obtained using the \verb!Build-gadget! procedure. 
If that fails (the second case), we will conclude the existence of a family of gadgets with observable gaps that correspond to a continuous
function. Then we will argue that the function must satisfy a certain functional
equation which will lead to a contradiction, and therefore will yield the gadgets we require, using the \verb!Build-gadget! procedure.

\begin{lemma}[Case I, \mbox{\cite[Lemma 23]{averages}}]\label{case1}
Suppose there exists $x\in I$ such that for some $t_1,t_2$ there exist
$(f_1,g_1)\in {\cal F}_{x,t_1}$ and $(f_2,g_2)\in {\cal F}_{x,t_2}$
such that $g_1(I'\times J)$ and $g_2(I'\times J)$ are disjoint. Let $\hat{O}$ be the distance of
$g_1(I'\times J)$ and $g_2(I'\times J)$. 

Then, there is an algorithm which, on input a rational $r\in (0,1/2)$, outputs in time  $poly(\bit(r))$
a pair of field gadgets ${\cal T}_1$ and ${\cal T}_2$, each of maximum degree 3 and size $O(|\log r|)$, such that 
\[|R_{\Tc_1} - x|, |R_{\Tc_2} - x|\leq r, \mbox{ but } |O_{\Tc_1}-O_{\Tc_2}|\geq\hat{O}/2.\]
\end{lemma}

\begin{lemma}[Case II, \mbox{\cite[Lemma 24]{averages}}]\label{case2}
Suppose that for every $x$, every $t_1,t_2$ and every two functions
$(f_1,g_1)\in {\cal F}_{x,t_1}$ and $(f_2,g_2)\in {\cal F}_{x,t_2}$
we have that $f_2(I'\times J)$ and $g_2(I'\times J)$ intersect. Then there exists a continuous
function $F:I\rightarrow J$ such that for every $x\in I$, every $\eps>0$ there exists $t_0$ such that for every $t\geq t_0$ and
every $(f,g)\in {\cal F}_{x,t}$ and every $R\in I'$ and every $O\in J$ we have
\begin{equation}\label{close}
|g(R,O) - F(x)|\leq\eps.
\end{equation}
\end{lemma}

The following lemma can be derived from \cite{averages} with minor adaptations.
\begin{lemma}[\mbox{\cite[Lemma 25]{averages}}]\label{lem:auxconstr}
Suppose \emph{Case II} happens, that is, the assumption of Lemma~\ref{case2} is satisfied, and let  $F$ be the continuous function 
guaranteed by Lemma~\ref{case2}. Suppose that there exist $x_1,x_2\in I$ such that the following equation is violated.
\begin{equation}\label{eq:functionFF}
F\left(\tfrac{1+\gamma\lambda x_1 x_2}{\beta+\lambda x_1x_2}\right) = \theta\big(\tfrac{1+\gamma\lambda x_1 x_2}{\beta+\lambda x_1x_2}\big)  - \omega\big(\tfrac{1+\gamma\lambda x_1 x_2}{\beta+\lambda x_1x_2}\big)\big(F(x_1) + F(x_2)\big).
\end{equation}
Then, there is an algorithm that computes rational numbers $x_1,x_2,\hat{R}, \hat{O}\in I$  in constant time, where the constant depends on $\beta,\gamma,\lambda,a,b,c, k$,   such that \eqref{eq:functionFF} is violated. Moreover, the algorithm, on input a rational $r\in (0,1/2)$, outputs in time $poly(\bit(r))$
a pair of field gadgets $\Tc_1$ and $\Tc_2$, each of maximum degree 3 and size $O(|\log r|)$, such that $|R_{\Tc_1} - \hat{R}|\leq r$,
$|R_{\Tc_2} - \hat{R}|\leq r$, but $|O_{\Tc_1}-O_{\Tc_2}|\geq\hat{O}$. 
\end{lemma}

Now assume that equation~\eqref{eq:functionFF} is satisfied for all $x_1,x_2\in I$. We are
going to derive a contradiction, thereby showing that ~\eqref{eq:functionFF} must be violated. Suppose $x_1,x_2,x_3\in I$ are such that $x_1 x_2 = x_3 x^*$. Plugging into~\eqref{eq:functionFF}
we obtain that $F$ has to satisfy the following functional equation
\[F(x_1) + F(x_2) = F(x_3) + F(x^*).\]
Using Cauchy's functional equation, the following is shown in \cite{averages}.

\begin{lemma}[\mbox{\cite[Lemma 25]{averages}}]\label{lem:3deecer}
Suppose $F$ is a continuous function on $I$. Suppose that for $x_1,x_2,x_3\in I$ such that
$x_1 x_2 = x_3 x^*$ we have $F(x_1) + F(x_2) = F(x_3) + F(x^*)$. Then there exists $K$ such that for all $x\in I$
we have $F(x) = K\log(x/x^*) + F(x^*)$.
\end{lemma}

Finally we show that~\eqref{eq:functionFF} has to be violated.

\begin{lemma}\label{final2spin}
A solution of the form $F(x) = K\log(x/x^*) + L$ for some constants $K,L$ cannot satisfy~\eqref{eq:functionFF}.
\end{lemma}

\begin{proof}
We will plug-in the solution $F(x) = K\log(x/x^*) + L$ into~\eqref{eq:functionFF} with $x_1=x^*$ and $x_2=y$. Let 
\[r(y)=\frac{1+\gamma\lambda x^* y}{\beta+\lambda x^* y},\mbox{ and }\hat{\theta}(R)=\theta(R)+a \cdot \omega(R)=- b \frac{\beta(R-\gamma)}{1-\beta \gamma}+c\frac{\gamma/R-\beta \gamma}{1-\beta\gamma}.\] 
Recall also that $\omega(R)= \frac{1+\beta \gamma-\beta R-\gamma/R}{1-\beta \gamma}$.

 We obtain
\begin{equation}\label{eq:mainfun}
K\log r(y) - K\log x^* + L = \hat{\theta}(r(y))-\omega(r(y))(K\log(y) - K\log(x^*)+a+  2 L).
\end{equation}
We will henceforth consider derivatives with respect to $y$. Observe that 
\[(K \log r(y))'=- K \frac{\lambda(1-\beta \gamma) x^*}{(\beta +\lambda x^* y)(1+\gamma\lambda x^* y)}=-\omega(r(y))(K\log(y))',\]
and therefore differentiating \eqref{eq:mainfun}, after cancellations, yields that
\begin{align*}
(\hat{\theta}(r(y)))'=(\omega(r(y)))'
( K\log(y) - K\log(x^*)- 1 +  2 F(x^*))
\end{align*}
We have 
\[\frac{\hat{\theta}'(R)}{\omega'(R)}= \frac{(- b \beta(R-\gamma)+c\gamma(1/R-\beta))'}{(1+\beta \gamma-\beta R-\gamma/R)'}=\frac{-b \beta-c\gamma/R^2}{-\beta+\gamma/R^2}\]
and therefore we obtain that
\begin{equation}\label{eq:qazdwd345}
\frac{-b \beta (r(y))^2-c\gamma}{-\beta (r(y))^2+\gamma}= K\log(y) - K\log(x^*)+  a+2 L.
\end{equation}
Differentiating once more with respect to $y$, we obtain that
\[\frac{2 \beta  \gamma \lambda x^* (b + c)(\beta + \lambda x^* y) (1 + 
   \gamma \lambda x^* y)}{(1 - \beta \gamma) (\beta - 
   \gamma \lambda^2 (x^*)^2 y^2)^2}=\frac{K}{y}\]
which can only be satisfied if (i) $\beta \gamma=K=0$, or (ii) $b+c=K=0$. For case (i), assume w.l.o.g. that $\gamma=0$. From \eqref{eq:qazdwd345}, we obtain $L=(b-a)/2$, and therefore $F(x)=(b-a)/2$. This can only satisfy \eqref{eq:functionFF} for all $x_1,x_2$ when $a=b=0$, in which case the observable is trivial (which is excluded). In case (ii), we may further assume that $\beta \gamma \neq 0$, otherwise we can conclude as before that the observable is trivial. From \eqref{eq:qazdwd345}, we obtain analogously that  $F(x)=(b-a)/2$ which can only satisfy \eqref{eq:functionFF} for all $x_1,x_2$ when $a=b=c=0$, which is again a trivial observable.
\end{proof}

Finally, we combine the various pieces to prove Lemma~\ref{lem:tcplus} and Theorem~\ref{thm:observablegadget}. We start with the proof of Theorem~\ref{thm:observablegadget}.
 
\begin{proof}[Proof of Theorem~\ref{thm:observablegadget}]
There are two cases to consider: either that of Lemma~\ref{case1} or that of Lemma~\ref{case2}. In Lemma~\ref{case1} applies, we are done. If on the other hand Lemma~\ref{case2} applies, by Lemmas~\ref{lem:3deecer} and~\ref{final2spin}, we obtain that \eqref{eq:functionFF} is violated, and hence Lemma~\ref{case2} yields the desired algorithm. Note that the observable gaps returned by these algorithms all lie in the interval $J$, cf. Lemma~\ref{lem:obsbounds}, which is bounded by absolute constants, finishing the proof. 
\end{proof}
\begin{proof}[Proof of Lemma~\ref{lem:tcplus}]
The lemma largely follows from the same proof as Theorem~\ref{thm:observablegadget}, the only extra ingredient that is needed is to lower-bound the difference between $R_{\Tc_\pl}$ and $R_{\Tc_\mi}$, which can be done by invoking the $\verb!Build-gadget!(x,t)$ procedure in Lemmas~\ref{case1} and~\ref{case2} for $x=\hat R+r/10$ and $x=\hat R+r/2$. The fact that we can ensure that $\hat R\neq 1$ follows from the assumption that at least one of $\beta\neq \gamma$ and $\lambda\neq 1$ holds.
\end{proof}

We can also give the proof of Lemma~\ref{lem:suscgadget} for the Potts model.
\begin{proof}[Proof of Lemma~\ref{lem:suscgadget}]
This is analogous to the proof of Theorem~\ref{thm:observablegadget}. The functional equation has the same form and the analogue of \eqref{eq:functionFF} is given by
\begin{equation}\label{eq:functionFF2}
F\left(\tfrac{1+\hat\gamma\hat\lambda x_1 x_2}{\hat\beta+\hat\lambda x_1x_2}\right) = \theta\big(\tfrac{1+\hat\gamma\hat\lambda x_1 x_2}{\hat\beta+\hat\lambda x_1x_2}\big)  - \omega\big(\tfrac{1+\hat\gamma\hat\lambda x_1 x_2}{\hat\beta+\hat\lambda x_1x_2}\big)\big(F(x_1) + F(x_2)\big),
\end{equation}
where $\omega(\cdot),\theta(\cdot)$ are the functions in Lemma~\ref{lem:Pottsrecursion}. It remains to establish that \eqref{eq:functionFF2} cannot have a solution of the form $F(x) = K\log(x) + L$ (note that for Potts, the value of $x^*$ is equal to 1). This follows by the same argument as in Lemma~\ref{final2spin}. Namely fix $x_1$ and plug in $y=x_2$. With $r(y)=\frac{1+\hat\gamma\hat\lambda x_1y}{\hat\beta+\hat\lambda x_1 y}$, we obtain
\begin{equation}\label{eq:mainfunpotts}
K\log r(y) + L = \theta(r(y))-\omega(r(y))(K\log(y) +K\log x_1+ 2 L).
\end{equation}
We will consider again derivatives with respect to $y$, and reuse the observation that
\[(K \log r(y))'=- K \frac{\lambda(1-\hat\beta \hat\gamma) x_1}{(\hat\beta +\hat\lambda x_1 y)(1+\hat\gamma\hat\lambda x_1 y)}=-\omega(r(y))(K\log(y)+K\log x_1+2L)'.\]
Differentiating \eqref{eq:mainfunpotts} therefore yields after cancellations that
\begin{align*}
(\theta(r(y)))'=(\omega(r(y)))'%\left(-\frac{(1-\beta \gamma)\lambda x^*(-\beta+\gamma \lambda^2 (x^*y)^2) }{(1+\gamma\lambda x^* y)^2(\beta+\lambda x^* y)^2} \right)
( K\log(y)+K\log x_1+2L)
\end{align*}
We have 
\[\frac{\theta'(B)}{\omega'(B)}= \frac{2 \beta (\hat \beta \hat \gamma-1) (B^2 + q-1)}{(\beta-1) (B^2 \hat \beta - 
  \hat \gamma) (\beta + q-1)}\]
and therefore 
\[\frac{2 \beta (\hat \beta \hat \gamma-1) ((r(y))^2 + q-1)}{(\beta-1) ((r(y))^2 \hat \beta - 
  \hat \gamma) (\beta + q-1)}= K\log(y) + K\log x_1+ 2 L.\]
Differentiating once more with respect to $y$ yields that $K$ must depend on $y$ which is a contradiction.
\end{proof}

\end{document}